\documentclass[12pt]{article}

\usepackage[a4paper,top=2.5cm,bottom=2.5cm,left=2.5cm,right=2.5cm,marginparwidth=1.75cm]{geometry}
\usepackage{lineno}
\usepackage{natbib}
\usepackage{amssymb}
\usepackage{bm}
\usepackage{authblk}
\usepackage{amsmath}
\usepackage{amsthm}
\usepackage{epsfig}
\usepackage{graphicx}
\usepackage{graphics}
\usepackage{float}
\usepackage{subfigure}
\usepackage{multirow}
\usepackage{color}
\usepackage{fullpage}
\usepackage[normalem]{ulem} 
\usepackage{makeidx}
\usepackage{xspace}
\usepackage{wrapfig}
\usepackage{setspace}
\usepackage{kotex}
\makeindex

\newtheorem{theorem}{Theorem}

\newtheorem{Algorithm}{Algorithm}

\newtheorem{remark}{Remark}

\newcommand{\xb}{{\bm{x}}}
\newcommand{\Xb}{{\bm{X}}}

\newcommand{\ub}{{\bm{u}}}

\newcommand{\betab}{{\bm{\beta}}}

\newcommand{\sigmab}{{\bm{\sigma}}}

\newcommand{\Sigmab}{{\bm{\Sigma}}}
\newcommand{\thetab}{{\bm{\theta}}}

\providecommand{\keywords}[1]
{
  \small	
  \textbf{\textit{Keywords: }} #1
}

\begin{document}

\title{Adaptive Accelerated Failure Time modeling with a Semiparametric Skewed Error Distribution}

\author[1]{Sangkon Oh}
\author[2]{Hyunjae Lee}
\author[3]{Sangwook Kang}
\author[4]{Byungtae Seo}
\affil[1]{Department of Statistics, Ewha Womans Univerisity}
\affil[2]{Department of Statistics, University of Pittsburgh}
\affil[3]{Department of Applied Statistics, Yonsei University}
\affil[4]{Department of Statistics, Sungkyunkwan University}

\date{}

\maketitle

\begin{abstract}
The accelerated failure time (AFT) model is widely used to analyze relationships between variables in the presence of censored observations. However, this model relies on some assumptions such as the error distribution, which can lead to biased or inefficient estimates if these assumptions are violated. In order to overcome this challenge, we propose a novel approach that incorporates a semiparametric skew-normal scale mixture distribution for the error term in the AFT model. By allowing for more flexibility and robustness, this approach reduces the risk of misspecification and improves the accuracy of parameter estimation. We investigate the identifiability and consistency of the proposed model and develop a practical estimation algorithm. To evaluate the performance of our approach, we conduct extensive simulation studies and real data analyses. The results demonstrate the effectiveness of our method in providing robust and accurate estimates in various scenarios.

\end{abstract}

\section{Introduction }
\label{sec1}
The accelerated failure time (AFT) model is widely used to analyze the relationship between covariates and a log-transformed failure time when there exist some censored observations. One notable advantage of the AFT model is that covariates directly affect the failure time, facilitating an intuitive interpretation of their impact and enabling predictions. If the distribution of the error is explicitly specified, the corresponding AFT model is referred to as a parametric AFT model.
Since parametric AFT models rely on the assumption regarding the underlying distribution of the error such as the normal, logistic or extreme value distributions, the deviation from the assumption for the error may lead bias or incorrect conclusions in the estimated parameters and subsequent inference.

For an alternative of parametric models, a semiparametric AFT model which does not require the error distribution can be employed. Commonly, this model is estimated using rank-based estimators, derived from the partial likelihood score function of the proportional hazards model \citep{prentice1978linear}. 
Another widely used approach for estimating regression coefficients in the semiparametric AFT model is the least squares method proposed by \cite{buckley1979linear}. This method replaces censored observations with their conditional expectations based on available information.

A limitation of the aforementioned methods for the semiparametric AFT model is the requirement to estimate the intercept separately from the other coefficients, which sets them apart from the parametric AFT model. Estimators for the intercept in these methods are not always consistent and may lead to biased estimates, despite the importance of the intercept in determining mean failure time and predicting failure times \citep{ding2015estimating}.
 In contrast, \cite{AFT_NGSM} suggested assuming a nonparametric Gaussian scale mixture distribution for the error and proposed the estimation method within the likelihood framework. This approach preserves the advantage of parametric models by enabling direct estimation of the intercept without the need for a separate procedure, while allowing for flexibility in modeling error distributions as rank-based estimators and least squares methods.

Nevertheless, it fails to capture asymmetric unimodal distributions although the class of the Gaussian scale mixture distributions covers a wide range of symmetric unimodal distributions. To overcome this issue, \cite{mattos2018likelihood} explored the use of some members in the class of skew-normal scale mixture distributions for modeling the error in censored regression, and \cite{ferreira2022linear} incorporated it into linear mixed models. Although these approaches address the issue of misspecification when the error follows whether the symmetric or asymmetric distribution, they are constrained to employing particular members from the class of skew-normal scale mixture distributions. Consequently, the methods require model selection from within this limited set of specific members, and there is still a potential for encountering misspecification issues.

In this paper, we present a novel approach by introducing a semiparametric skew-normal scale mixture distribution (SSNSM) for the error in the AFT model. The SSNSM distribution offers greater flexibility by accommodating a wide range of distributions, including both the class of the Gaussian scale mixture distributions and important asymmetric unimodal distributions such as skew-normal, skew-t and skew-slash distributions, without any model selection procedure. While some distributions may not strictly belong to the class of skew-normal scale mixture distributions, we can expect that the flexibility of SSNSM minimizes a potential misspecification problem. Moreover, the stability of intercept estimation is maintained as it is directly estimated within the likelihood framework.
Recently, \cite{finitemixture_ssnsm} proposed a finite mixture model with multivariate SSNSM for each component distribution, and demonstrated the superiority of the SSNSM in capturing the characteristics of each component distribution.

The remainder of this paper is organized as follows. Section 2 and Section 3 review AFT models and the SSNSM distribution, respectively. Section 4 presents the proposed model along with estimation procedure. Simulation studies are presented in Section 5, while applications to real-world datasets are discussed in Section 6. Some concluding remarks are provided in Section 7.

\section{Accelerated failure time models}
\label{sec2}

\subsection{Parametric AFT model}
The AFT model is a useful model for directly investigating the relationship between log-transformed failure times and covariates, which can be represented as
\begin{align}
\label{aft}
\log T = \beta_0 + \xb^\top \betab + \epsilon,
\end{align}
where $T$ denotes the potential failure time, $\xb$ is a $p$-dimensional vector of covariates, $(\beta_0, \betab^{\top})^{\top}$ represents a $(p+1)$-dimensional vector of regression coefficients, and $\epsilon$ is an error term satisfying $E(\epsilon)=0$. In the presence of right censoring, the observed time is defined as $Y = \min(T,C)$, where $C$ is the potential censoring time. Additionally, an indicator variable $\delta$ is introduced, denoted as $\delta = I(T \leq C)$, which takes the value $\delta = 1$ if $T$ is observed and $\delta = 0$ otherwise. We further assume that $T$ and $C$ are conditionally independent given the covariates $\xb$. The observed data can be represented as $(y_i, \delta_i, \xb_i^{\top})$, $i=1,2,\ldots,n$, where all samples are independent and identically distributed and $n$ represents the sample size.

Let $f(\cdot)$ and $F(\cdot)$ denote the probability density function (pdf) and cumulative distribution function (cdf) of $\epsilon$, respectively. The corresponding survival function is denoted as $S(\cdot) = 1 - F(\cdot)$. Then, for a given sample of $n$ observations, the likelihood function based on \eqref{aft} can be expressed as
\begin{align*}
L(\thetab) = \prod_{i=1}^{n} f(\epsilon_i ; \thetab)^{\delta_i} S(\epsilon_i ; \thetab)^{1-\delta_i},
\end{align*}
where $\thetab$ represents all the parameters in the model. By maximizing this likelihood function with respect to $\thetab$, we can obtain the maximum likelihood estimator (MLE) for $\thetab$. A common choice for the distribution of $\epsilon$ is a normal density, but other distributions can also be considered. Although parametric AFT models are advantageous in computing the estimate of $\thetab$ and in studying theoretical properties for the estimators, the unverifiable parametric assumption poses a risk of misspecification, as the assumed distribution may not accurately reflect the true underlying distribution of the failure time. In this case, the estimator may produce a seriously biased estimator and experience a substantial efficiency loss. 

\subsection{Semiparametric approaches}
As an alternative to parametric AFT models, a semiparametric AFT model which does not require a specification of error distribution has been proposed. One popular approach for estimating the semiparametric AFT model is a rank-based estimator \citep{prentice1978linear}. The rank-based estimator with a Gehan-type weight function has estimating functions given by 
\begin{align}
U_g(\betab) = \sum_{i=1}^{n} \sum_{j=1}^{n} \delta_i ( \xb_i - \xb_j) I[\epsilon_j(\betab) \geq \epsilon_i(\betab)], \label{gehan}
\end{align}
where $\epsilon_i(\betab) = \log y_i - \xb_i^{\top} \betab$ for $i=1,2,\dots,n$.
To address computational challenges arising from lack of smoothness in \eqref{gehan},  \cite{brown2007induced} proposed an induced smoothing method with the Gehan-type weight. The induced smoothed estimator can be obtained using
\begin{align*}
U_s(\betab) &= \sum_{i=1}^{n} \sum_{j=1}^{n}   \delta_i ( \xb_i - \xb_j)  \Phi \Bigg [n^{1/2} \Bigg \{ \frac{\epsilon_j(\betab) -\epsilon_i(\betab)}{(\xb_i - \xb_j)^{\top} \Sigmab (\xb_i - \xb_j)}  \Bigg \} \Bigg ],
\end{align*}
where $\Phi(\cdot)$ denotes the cdf of the standard normal distribution, and $\Sigmab$ indicates a symmetric and positive definite matrix with dimensions $p \times p$.
These rank-based estimators are consistent and asymptotically normal under certain regularity conditions, providing reliable estimates of the regression coefficients $\betab$ (\citealp{tsiatis1990estimating}; \citealp{ying1993large}).

Another popular estimator is Buckley-James estimator \citep{buckley1979linear} that is based on a modified least-squares method.  This estimator is obtained by solving the equations
\begin{align}
\label{BJ}
\sum_{i=1}^{n} (\xb_i - \bar{\xb}) (\log \hat{Y}_i - \xb_i^{\top} \betab) = 0,
\end{align}
where $\bar{\xb}$ represents the mean of the covariates and $\log \hat{Y}_i$ is the estimated expected value of $\log T_i$ given $Y_i$, $X_i$ and $\delta_i$.
However, solving equations \eqref{BJ} directly for $\betab$ is numerically unstable which often suffers from non-convergence issues. To address this issue, \cite{jin2006least} proposed an iterative approach based on generalized estimating equations (GEE).
The GEE approach involves solving the equations by iteratively updating between $\bm{b}$ and $\betab$:
\begin{align*}
\sum_{i=1}^{n} (\xb_i - \bar{\xb}) \left ( \log \hat{Y}_i (\bm{b}) - n^{-1} \sum_{i=1}^{n} \log \hat{Y}_i (\bm{b}) - ( \xb_i - \bar{\xb})^{\top} \betab \right) = 0,
\end{align*}
which ensures convergence with a consistent initial estimator. 

In aforementioned approaches for the semiparametric AFT model, estimating the intercept is typically handled separately from estimating the other coefficients. \cite{ding2015estimating} proposed the method of estimating the intercept as $\hat{\beta}_0 = \int_{-\infty}^{\infty} t d\hat{F}(t;\hat{\betab})$, where $\hat{F}(t;\hat{\betab})$ is the Kaplan-Meier estimator of the distribution function of $\epsilon$ given estimates ${\hat{\betab}}$. Under certain regularity conditions, such as the unboundedness of the support of $\xb$, this estimator has been shown to be consistent for $\beta_0$.
However, in practical scenarios where the covariate support is narrow and bounded, the estimator of $\beta_0$ may not perform well. 

A recent study by \cite{AFT_NGSM} proposed modeling the error term $\epsilon$ in \eqref{aft} using the nonparametric Gaussian scale mixture distribution. This approach offers an advantage over rank-based estimators and least squares methods as it enables direct estimation of the intercept without requiring a separate procedure. Furthermore, the estimator for $(\beta_0, \betab^{\top})^{\top}$ using the nonparametric Gaussian scale mixture error distribution is the MLE, and has been shown to be consistent.
However, it may still face the issue of misspecification because the class of Gaussian scale mixture distributions cannot encompass asymmetric unimodal distributions, despite including many essential symmetric unimodal distributions. To address this limitation, we introduce the SSNSM distribution  as an alternative to the nonparametric Gaussian scale mixture distribution in the following section.

\section{Semiparametric skew normal scale mixtures}
\label{sec3}

To address the need for a distribution that can demonstrate both symmetric and asymmetric data, a skew-normal distribution was introduced by \citet{azzalini1985class} as
\begin{align}
f(\epsilon; \xi, \lambda, \sigma) =  \frac{2}{\sigma} \phi\left(\frac{\epsilon - \xi}{\sigma}\right)  \Phi\left(\lambda \frac{\epsilon - \xi }{\sigma}\right), 
\label{sn}
\end{align}
where $\phi(\cdot)$ and $\Phi(\cdot)$ represent the pdf and cdf of the standard normal distribution, respectively. The parameters of the distribution include $\xi$ as the location parameter, $\lambda$ as the slant parameter and $\sigma$ as the scale parameter. When $\lambda$ is equal to zero, the skew-normal density turns into the standard normal density. The magnitude of the skewness is controlled by the value of $\lambda$, with positive values indicating right skewness and negative values indicating left skewness.

In order to accommodate heavier tails in comparison to those of the skew-normal distribution \eqref{sn}, a skew-normal scale mixture distribution was introduced by \cite{branco2001general}. 
When $\epsilon$ follows a skew-normal scale mixture distribution, denoted as $\text{SNSM}(\xi, Q, \lambda)$, the density function of this distribution is given by
\begin{align}
f(\epsilon; \xi, \lambda, Q)=\int_{0}^{\infty} \frac{2}{\sigma} \phi\left(\frac{\epsilon - \xi }{\sigma}\right)  \Phi\left(\lambda \frac{\epsilon - \xi}{\sigma}\right) dQ(\sigma),  \label{SSNSM}
\end{align}
where $Q$ represents the latent distribution function defined on $\mathbb{R}^+=(0,\infty)$. The class of skew-normal scale mixture distributions encompasses not only the class of Gaussian scale mixture distributions containing the normal, t and Laplace, but also skewed distributions such as the skew-normal, skew-t, skew-slash, skew-contaminated normal and others. The previous studies have been successfully applied to model error distributions in various statistical models. These studies include the use of \eqref{SSNSM} in robust mixture of regressions by \cite{zeller2016robust}, its application in censored regression by \cite{mattos2018likelihood}, and its incorporation into linear mixed models by \cite{ferreira2022linear}. Note that these methods assume parametric forms for the distribution of $Q$, which requires additional model selection problem. 

However, assuming a specific parametric form for $Q$ in \eqref{SSNSM} can lead to misspecification issues if the true underlying distribution deviates from the assumed form. Moreover, selecting the appropriate parametric model for $Q$ requires model selection procedures, which can be challenging. To address these limitations, it is advantageous to leave $Q$ unspecified, allowing for greater flexibility in modeling a wide range of distributions. By not making explicit assumptions about the form of $Q$, the SSNSM distribution provides automatic adaptability to various types of data. This flexibility enables the SSNSM to achieve robustness and adaptiveness without the need for explicit model selection procedures.

With parametric $Q$, the estimation of parameters can be achieved by using some standard optimization algorithms such as Newton-type algorithms and the expectation-maximization algorithm. However, with unspecified distribution $Q$, those algorithms cannot be directly applied because we need to maximize the likelihood over infinite dimensional space. Fortunately, there exist several algorithms based on the directional derivative, including the vertex direction method \citep{Bohning85}, the vertex exchange method \citep{Bohning86}, the intra-simplex direction method \citep{LK92} and constrained Newton method for multiple support points (CNM; \citealp{Wang07}). These algorithms are employed to compute the nonparametric maximum likelihood estimator (NPMLE) of $Q$ in \eqref{SSNSM}.

To explain this, consider a random sample $\epsilon_1, \ldots, \epsilon_n$ drawn from the SSNSM distribution. Denote the log-likelihood of $Q$ as $\tilde{\ell}_n(Q) = \sum_{i=1}^n \log f(\epsilon_i; \xi, \lambda, Q)$, where $\xi$ and $\lambda$ are fixed. The NPMLE of $Q$ can be characterized using the directional derivative of $\tilde{\ell}_n(\hat Q)$ from $\hat Q$ (current estimator) to a unit step distribution function at $\sigma$, denoted by $H_\sigma$:
\begin{align*}
D_{\hat Q}(\sigma) = \lim_{\alpha \to 0} \frac{\tilde{\ell}_n\{(1-\alpha){\hat Q} + \alpha H_\sigma\} - \tilde{\ell}_n({\hat Q})}{\alpha}.
\end{align*}
If $D_{\hat Q}(\sigma) > 0$, it implies that there exists some $0 < \alpha \leq 1$ such that $\tilde{\ell}_n\{(1-\alpha){\hat Q} + \alpha H_\sigma\} > \tilde{\ell}_n({\hat Q})$. In other words, the current estimator ${\hat Q}$ is not the NPMLE, and updating the estimator with $(1-\alpha){\hat Q} + \alpha H_\sigma$ can lead to bigger likelihood than before.
The necessary and sufficient condition for $\hat{Q}$ to serve as the NPMLE of $Q$ is that $D_{\hat Q}(\sigma) \leq 0$ holds for all $\sigma \in \mathbb{R}^+$, and $D_{\hat Q}(\sigma^*) = 0$ for all support points $\sigma^*$ in $S(\hat{Q})$, where $S(\hat{Q})$ denotes the set of support points of $\hat{Q}$ \citep{Lindsay_1995}.
The NPMLE has proven to be a valuable tool in a variety of statistical models, such as the linear regression problem (\citealp{Seo_2017}), generalized autoregressive conditional heteroscedasticity model (\citealp{SL15}), AFT model \citep{AFT_NGSM}, finite mixture models (\citealp{XYS16}; \citealp{finitemixture_ssnsm}) and finite mixture of regressions \citep{FMR_NGSM}.

\section{AFT with the SSNSM error distribution}
\label{sec4}

\subsection{Proposed model}
Under \eqref{aft} and \eqref{SSNSM}, the pdf of $\epsilon$ is given by
\begin{align}
\label{aft_ssnsm}
f(\epsilon;\thetab, Q) = \int_{0}^{\infty} \frac{2}{\sigma} \phi\left(\frac{\epsilon }{\sigma}\right)  \Phi\left(\lambda \frac{\epsilon }{\sigma}\right) dQ(\sigma), 
\end{align}
where $\thetab = (b_0, \betab^{\top}, \lambda)^{\top}$ and $\epsilon = \log T - b_0 - \xb^{\top} \betab$. The parameter $b_0$ represents the location parameter for $\epsilon$ following $\text{SNSM}(0, Q, \lambda)$, while $\betab$ denotes the regression coefficients for slopes.
 $\lambda$ indicates the slant parameter, and $Q$ represents the latent distribution of the scale parameter $\sigma$.
For accurate estimates of the mean failure time, the intercept term $\beta_0$ is computed from $\beta_0 = b_0 + \sqrt{\frac{2}{\pi}} \frac{\lambda}{\sqrt{1 + \lambda^2}} \int_{0}^{\infty} \sigma d Q(\sigma)$ because    
 \begin{align*}
E[\epsilon \mid \xb] =  \sqrt{\frac{2}{\pi}} 
 \frac{\lambda}{\sqrt{1 + \lambda^2}} \int_{0}^{\infty} \sigma d Q(\sigma).
\end{align*}
A similar approach has been utilized in prior studies such as \cite{zeller2016robust}, \cite{mattos2018likelihood} and \cite{ferreira2022linear} to derive accurate estimates of the mean failure time.

\cite{mattos2018likelihood} proposed to use some specific parametric distributions of $Q$ to construct skew-normal, skew-t, skew-slash, and skew-contaminated normal distributions. This approach is beneficial when the error distribution is known or can be reasonably assumed based on domain expertise. However, such an approach can result in the misspecification problem, which can lead to biased or inconsistent estimates. In contrast, the proposed method with SSNSM error distribution \eqref{aft_ssnsm} offers a high level of flexibility by not specifying the latent distribution $Q$, eliminating the need for a specific parametric form. While the proposed method is not completely immune to misspecification, it is reasonable to assume that the error distribution aligns with or closely resembles the class of skew-normal scale mixture distributions. Additionally, the likelihood approach for estimating all coefficients provides an advantage over methods that require a separate procedure for intercept estimation. Note that the proposed method incorporates all the advantages of the \cite{AFT_NGSM} approach because the class of skew-normal scale mixture distributions contains that of Gaussian scale mixture distributions.

The identifiability of \eqref{aft_ssnsm} could be an issue becasue $Q$ is unspecified in our model. We show this identifiability in Theorem \ref{thm1}.
\begin{theorem}
    \label{thm1}
     Suppose that the support of $\xb$ contains an open set in $\mathbb{R}^{p}$. If 
    \begin{align*}
    \int_{0}^{\infty} \frac{2}{\sigma} \phi\left(\frac{\log t - b_0 - \xb^{\top}\betab }{\sigma}\right)  \Phi\left(\lambda \frac{\log t - b_0 - \xb^{\top}\betab }{\sigma}\right) dQ(\sigma) \\
    = \int_{0}^{\infty} \frac{2}{\sigma} \phi\left(\frac{\log t - \tilde{b}_0 - \xb^{\top}\tilde{\betab} }{\sigma}\right)  \Phi\left(\tilde{\lambda} \frac{\log t - \tilde{b}_0 - \xb^{\top}\tilde{\betab} }{\sigma}\right) d\tilde{Q}(\sigma),
    \end{align*}
    then $(b_0, \betab^{\top}, \lambda, Q)^{\top}$ is equal to $(\tilde{b}_0, \tilde{\betab}^{\top}, \tilde{\lambda}, \tilde{Q})^{\top}$ for almost all $t$.
\end{theorem}
\begin{proof}
A proof is given in the Appendix A. 
\end{proof}

The general consistency of structural parameters and latent distribution in semiparametric mixture distributions was established by \cite{kiefer1956consistency}. Based on their results, \cite{AFT_NGSM} demonstrated the consistency of AFT model with the nonparametric Gaussian scale mixture error distribution. 
Now, we specifically address the consistency of MLE denoted as $(\hat{b}_0, \hat{\betab}^{\top}, \hat{\lambda}, \hat{Q})$ for $(b_0, \betab^{\top}, \lambda, Q)$ in \eqref{aft_ssnsm}.
Let $g(\cdot)$ and $G(\cdot)$ denote the density and distribution functions of $\log C$, respectively. The consistency for the MLE of $(b_0, \betab^{\top}, \lambda, Q)$ is summarized in Theorem \ref{thm2}.
\begin{theorem}
    \label{thm2}
    Assume that the support of $Q$ is a subset of $[\ell, \infty)$ for a positive constant $\ell$ and $-\int_{-\infty}^{\infty} \log g(t) d G(t) < \infty$.    
    If $\int_{\ell}^{\infty} \log \sigma d Q(\sigma) < \infty$, then
    the MLE of $(b_0, \betab^{\top}, \lambda, Q)^{\top}$ is consistent estimator.
\end{theorem}

\begin{proof}
A proof is given in the Appendix B. 
\end{proof}

\begin{remark}
$\hat{\beta}_0 = \hat{b}_0 + \sqrt{\frac{2}{\pi}} \frac{\hat{\lambda}}{\sqrt{1 + \hat{\lambda}^2}} \int_{0}^{\infty} \sigma d \hat{Q}(\sigma)$ is also consistent estimator for $\beta_0$ due to the invariance property of the MLE. 
\end{remark}

\subsection{Estimation}

When considering random right censoring, the log-likelihood based on the AFT model with SSNSM error distribution \eqref{aft_ssnsm} can be expressed as
\begin{align}
\label{loglik_ssnsm}
\ell(\thetab, Q) = \sum_{i=1}^{n} \log  \int_{0}^{\infty} \{ f(\epsilon_i ; \thetab, \sigma) \}^{\delta_i} \{ S(\epsilon_i ; \thetab, \sigma) \}^{1-\delta_i} dQ(\sigma),
\end{align}
where $f(\epsilon_i ; \thetab, \sigma)$ and $S(\epsilon_i ; \thetab, \sigma)$ are the probability density function \eqref{sn} and survival function of the skew-normal distribution, respectively. 
The survival function of the skew-normal distribution can be expressed as
\begin{align*}
S(\epsilon ; \thetab, \sigma) = 1 - \Phi \Bigg ( \frac{\epsilon}{\sigma} \Bigg ) + 2 T \Bigg  ( \frac{\epsilon }{\sigma}, \lambda \Bigg ), 
\end{align*}
where $T(h, \lambda)$ is $T$ function \citep{owen1956tables}.

We present an algorithm that iteratively maximizes the log-likelihood function \eqref{loglik_ssnsm} to obtain the MLE of $(\thetab, Q)$.
First, to update $Q$ while keeping current estimate $\hat{\thetab}$ for $\thetab$ fixed, we maximize the function
\begin{align*}
\tilde{\ell}_n(Q) = \sum_{i=1}^{n} \log  \int_{0}^{\infty}  f(\epsilon_i ; \hat{\thetab}, \sigma) ^{\delta_i}  S(\epsilon_i ; \hat{\thetab}, \sigma) ^{1-\delta_i} dQ(\sigma).
\end{align*}
The directional derivative of $\tilde{\ell}_n$ at $Q$ toward $H_{\sigma^*}$ is given by
\begin{align}
D_{Q}(\sigma^*) & = \lim_{\alpha \to 0} \frac{\tilde\ell_{n}\{(1-\alpha)Q + \alpha H_{\sigma^*}\} - \tilde{\ell}_n(Q)}{\alpha} \notag  \\
        &= \sum_{i=1}^{n}  \frac{   f(\epsilon_i ; \hat{\thetab}, \sigma^{*}) ^{\delta_i}  S(\epsilon_i ; \hat{\thetab}, \sigma^{*}) ^{1-\delta_i} }{\int  f(\epsilon_i ; \hat{\thetab}, \sigma) ^{\delta_i}  S(\epsilon_i ; \hat{\thetab}, \sigma) ^{1-\delta_i} dQ(\sigma) } - n,  \label{D} 
\end{align}
where $\sigma^* \in \mathbb{R}^+$. 
The NPMLE $\hat{Q}$ is obtained using the CNM algorithm which is based on \eqref{D}. The CNM algorithm is specifically selected due to its efficiency in discarding unnecessary or improper support points, resulting in a faster and more practical estimation process compared to alternative algorithms.
The CNM algorithm can be summarized as follows.
\begin{Algorithm}
\label{alg1}
    \item[Step 1.] Identify all local maximizers of $D_Q(\sigma)$ using the current estimates. If all local maximizers of \eqref{D} are less than or equal to $0$, the algorithm terminates. Otherwise, the current support point set $S({Q^{(t)}})$ is updated by adding local maximizers to create a new support point set $S(Q^{(t+\frac{1}{2})})$ with $K^*$ elements.
    
    \item[Step 2.] Based on support point set $S(Q^{(t+\frac{1}{2})})$,  find $\alpha^{(t+\frac{1}{2})}$ as
    \begin{align*}
    {\bm{\alpha}}^{(t+\frac{1}{2})} = \arg\min_{\bm{{\alpha}}}  |\bm{{\alpha}}^{\top} \bm{1} - 1 |^2  + \gamma \| \bm{J} \bm{{\alpha}}  - \bm{2} \|_2^2,
    \end{align*}
    subject to $\bm{{\alpha}} \geq \bm{0}$ and $\gamma > 0$, by non-negative least squares algorithm (Lawson and Hanson, 1974). 
    Here, $\bm{J}$ is a $n \times K^*$ matrix with the $(i,k^*)th$ element defined as 
    \begin{align*}
        \frac{f(\epsilon_i ; \hat{\thetab}, \sigma_{k^*})^{\delta_i} S(\epsilon_i ; \hat{\thetab}, \sigma_{k^*})^{1-\delta_i} }{\sum_{k=1}^{K^*} \alpha_k f(\epsilon_i ; \hat{\thetab}, \sigma_k)^{\delta_i} S(\epsilon_i ; \hat{\thetab}, \sigma_k)^{1-\delta_i}},
    \end{align*}

    \item[Step 3.] Update the set of support points $S(Q^{(t+1)})$ and the weights $\bm{\alpha}^{(t+1)}$ by removing the support points with zero weight from $S(Q^{(t+\frac{1}{2})})$, then return to step 1.
    
\end{Algorithm}

Next, we update the parameter vector $\thetab$ by maximizing the function
\begin{align}
\label{profile}
\ell(\thetab, \hat{Q}) = \sum_{i=1}^{n} \log \int_{0}^{\infty} { f(\epsilon_i ; \thetab, \sigma) }^{\delta_i} { S(\epsilon_i ; \thetab, \sigma) }^{1-\delta_i} d\hat{Q}(\sigma),
\end{align}
while keeping the estimate $\hat{Q}$ fixed. Given that the NPMLE of $Q$ is known to inherently have a finite number of support points \citep{Lindsay_1995}, we can reasonably assume that the NPMLE $\hat{Q}$ has $K$ support points without loss of generality. Let $\hat{\sigmab} = (\hat{\sigma}_{1}, \hat{\sigma}_{2}, \ldots, \hat{\sigma}_{K})^{\top}$ and $\hat{\bm{\alpha}} = (\hat{\alpha}_{1}, \hat{\alpha}_{2}, \ldots, \hat{\alpha}_{K})^{\top}$ represent the support points and corresponding weights of $\hat{Q}$, respectively. Then, \eqref{profile} can be expressed as 
\begin{align}
\label{profile2}
\ell(\thetab, \hat{Q}) = \sum_{i=1}^{n} \log  \sum_{k=1}^{K} \hat{\alpha}_k  f(\epsilon_i ; \thetab, \hat{\sigma}_k) ^{\delta_i}  S(\epsilon_i ; \thetab, \hat{\sigma}_k)^{1-\delta_i} ,
\end{align}
where $\hat{\alpha}_k > 0$ for $k = 1, 2, \cdots, K$ and $\sum_{k=1}^{K} \hat{\alpha}_k = 1$. 
Because there is no closed form solutions for the maximizer of \eqref{profile2}, a numerical method is necessary. In this paper, the Broyden-Fletcher-Goldfarb-Shanno (BFGS) algorithm (\citealp{broyden1970convergence}; \citealp{fletcher1970new}; \citealp{goldfarb1970family}; \citealp{shanno1970conditioning}) is employed to estimate the parameter $\boldsymbol{\theta}$ in \eqref{profile2} as it offers low computational complexity compared to the Newton-Raphson method.  Let $h(\boldsymbol{\theta})$ and $\bm{B}$ represent $-\ell(\boldsymbol{\theta}, \hat{Q})$ and the approximation of the Hessian matrix for $h(\boldsymbol{\theta})$, respectively. Then, this algorithm can be summarized as follows.
\begin{Algorithm}
\label{alg2}
   
    \item[Step 1.] Update $\boldsymbol{\theta}^{(t)}$ to $\thetab^{(t+1)}$ as    
    \begin{align*}
\boldsymbol{\theta}^{(t+1)} = \boldsymbol{\theta}^{(t)} - c \bm{a}^{(t)},
\end{align*}
where $\bm{a}^{(t)}$ is obtained by solving the equation $\bm{B}^{(t)} \bm{a}^{(t)} = - \nabla h(\boldsymbol{\theta}^{(t)})$, and $c$ is a positive value representing the step size. If the algorithm has converged, then it is stopped. Otherwise, proceed to step 2.

    \item[Step 2.] Update the approximation of the Hessian matrix $\bm{B}^{(t+1)}$ as
    \begin{align*}
        \bm{B}^{(t+1)} = \bm{B}^{(t)} - \frac{\bm{B}^{(t)} \bm{a}^{(t)} {\bm{a}^{(t)}}^{\top} \bm{B}^{(t)} }{{\bm{a}^{(t)}}^{\top} \bm{B}^{(t)} \bm{a}^{(t)} } + \frac{\bm{d}^{(t)} {\bm{d}^{(t)}}^{\top}}{{\bm{d}^{(t)}}^{\top} \bm{a}^{(t)}},
    \end{align*}
    where $\bm{d}^{(t)} = \nabla h(\boldsymbol{\theta}^{(t+1)}) - \nabla h(\boldsymbol{\theta}^{(t)})$, and return to step 1.
    
\end{Algorithm}

The alternating repetition between Algorithm \ref{alg1} and \ref{alg2} continues until there is no further increase in the log-likelihood. Subsequently, $\hat{\beta}_0$ is updated based on $\hat{b}_0$, $\hat{\lambda}$, and $\hat{Q}$.
In summary, the algorithm for estimating $(\thetab, Q)$ in \eqref{loglik_ssnsm} can be summarized as follows.
\begin{Algorithm}
\label{alg3}

    \item[step 1.] For fixed $\thetab$, find $Q$ using CNM algorithm.
    \item[step 2.] For fixed $Q$, find $\thetab$ using BFGS algorithm.
    \item[step 3.] Repeat step 1 and step 2 until the log-likelihood \eqref{loglik_ssnsm} does not increase.
    \item[step 4.] Update $\beta_0$ based on $b_0$, $\lambda$ and $Q$.

\end{Algorithm}

\section{Simulation studies}
\label{sec5}

In this section, we conduct a simulation study to evaluate the performance of the proposed method in comparison to other estimation methods. For each simulated sample, we employ the following estimators.
\begin{enumerate}
\item Normal: MLE assuming the normally distributed error.
\item SN: MLE assuming the skew-normal error.
\item Gehan:  Induced-smoothed rank-based estimator with Gehan-type weight.
\item GEE: Generalized estimating equation estimator.
\item SSNSM: MLE assuming the SSNSM error.
\end{enumerate}
To fit the Normal, we use the $\texttt{survival}$ package in \textsf{R} \citep{therneau2019package}. The $\texttt{aftgee}$ package in \textsf{R} \citep{chiou2014fitting} is utilized for fitting the Gehan and GEE.  
Additionally, we also develop a \textsf{R} program for fitting the SN and SSNSM.

\subsection{Simulation 1}
\label{sec5-1}

In the first simulation study, we assess the performance of the proposed method in comparison to other methods under various scenarios involving different error distributions, sample sizes and censoring proportions.
The log-transformed failure times $\log T_i$ ($i = 1,2, \cdots, n$) are generated according to the relationship as $\log T = 2 + X_1 - X_2 + \epsilon$, where $X_1$ is generated from a standard normal distribution and $X_2$ is generated from a Bernoulli distribution with a success probability of 0.5.
The log-transformed censoring times $\log C_i$ ($i = 1,2, \cdots, n$) are generated from a uniform distribution with a minimum of 0 and a maximum of $\tau$. Specifically, we use $\log C \sim U(0, \tau)$, where $\tau$ takes on the values of 1.5 or 4. These values correspond to censoring proportions of approximately $0.6$ to $0.8$ and $0.3$ to $0.5$, respectively.

We consider two sample sizes, $n=200$ and $n=400$, and explore the impact of four distinct error distributions for $\epsilon$. These error distributions include the standard normal distribution, Student's t distribution with 3 degrees of freedom, Gumbel distribution with location 0 and scale parameter 5, and skew-t distribution with location 0, scale 1, slant parameter -15 and 3 degrees of freedom.
Note that the Gumbel distribution does not fall within the class of skew-normal scale mixture distributions.
To ensure comparability among the distributions, we standardize them to have a zero mean and a unit variance. Additionally, in each simulation, we assume that the observation with the maximum among $\epsilon_i$, where $i = 1, 2, \ldots, n$, is always observed. 

The performance of each estimation method is assessed by computing the mean squared error (MSE) and bias for each parameter across $200$ replicated samples. The MSE and bias are calculated as follows.
\begin{align*}
MSE = \frac{1}{200} \sum_{j = 1}^{200} (\hat{\beta}_{k(j)} - {\beta}_{k})^2, \ \  bias = \frac{1}{200} \sum_{j = 1}^{200} (\hat{\beta}_{k(j)} - {\beta}_{k}), 
\end{align*}
where $\beta_{k}$ and $\hat{\beta}_{k(j)}$ are the $k$th true regression coefficient and the estimate of the $\beta_{k}$ from the $j$th sample for $k = 0,1,2$ and $j = 1,2,\ldots, 200$, respectively. 
Note that the values of $\beta_0$, $\beta_1$ and $\beta_2$ are 2, 1 and -1, respectively.
By examining the performance under the different settings, we can gain insights into the robustness and effectiveness of the proposed method for estimating the parameters in the AFT model.

\begin{table}[p]
\caption{MSE and bias when $n=200$ with $0.3$ to $0.5$ censoring rate (Boldfaced numbers indicate the smallest value in each method)}
\label{tab:n200,c30} 
	\centering
	\begin{tabular}{ c  c  c   c    c   c  } 
        \hline  \hline \noalign{\smallskip}
	\multirow{2}{*}{Method}	&
	\multirow{2}{*}{}	&
	$\text{N}(0,1)$ & $\text{t}(3)$ & $\text{Gumbel}(0,5)$ & $\text{skew-t}(0,1,-15,3)$ \\
 & & MSE  (bias) & MSE  (bias) & MSE  (bias) & MSE  (bias) \\ \hline
    \multirow{3}{*}{Normal} & $\beta_0$ &  $\boldsymbol{0.0140}$  (0.0165) & 0.0306  (0.0961) & 0.0124  (0.0225) & 0.0455  (0.1682) \\
                            & $\beta_1$ &  0.0082  (0.0091) & 0.0103  (0.0316) & 0.0074  (0.0224) & 0.0150   (0.0490) \\
                            & $\beta_2$ & 0.0264  (0.0021) & 0.0339  (-0.0423) & 0.0313  (-0.0417) & 0.0446  (-0.0777) \\
\hline
    \multirow{3}{*}{SN} & $\beta_0$ &  0.0145 (0.0215) & 0.0283  (0.0785) & 0.0100 (0.0442) & 0.0249 (-0.1243) \\
                            & $\beta_1$ & 0.0082  (0.0076) & 0.0096 (0.0305) & $\boldsymbol{0.0047}$ (0.0072) & 0.0048 (0.0412) \\
                            & $\beta_2$ & 0.0270 (0.0034) & 0.0299 (-0.0335) & $\boldsymbol{0.0181}$ (-0.0185) & 0.0132 (-0.0422) \\
\hline
    \multirow{3}{*}{Gehan} & $\beta_0$ &  0.0152  (0.0250) & 0.0229  (0.0647) & 0.0153  (0.0705) & 0.0092  (0.0216) \\
                           & $\beta_1$ &  $\boldsymbol{0.0080}$  (0.0039) & $\boldsymbol{0.0053}$  (0.0021) & 0.0049  (0.0106) & 0.0061  (0.0076) \\
                           & $\beta_2$ & 0.0270  (0.0027) & $\boldsymbol{0.0201}$ (-0.0029) & 0.0214  (-0.0203) & 0.0207   (-0.0229) \\
\hline
    \multirow{3}{*}{GEE} & $\beta_0$ &  0.0145  (0.0258) & 0.0298  (0.0733) & 0.0207  (0.0827) & 0.0185  (0.0319) \\
                         & $\beta_1$ &  0.0083  (0.0096) & 0.0090 (0.0107) & 0.0075  (0.0228) & 0.0140  (0.0045) \\
                         & $\beta_2$ & $\boldsymbol{0.0261}$  (0.0017) & 0.0330  (-0.0174) & 0.0331  (-0.0203) & 0.0427   (-0.0380) \\
          \hline
    \multirow{3}{*}{SSNSM} & $\beta_0$ &  0.0159 (0.0183) & $\boldsymbol{0.0212}$  (0.0305) & $\boldsymbol{0.0100}$  (0.0368) & $\boldsymbol{0.0060}$  (0.0281) \\
                           & $\beta_1$ &  0.0084  (0.0052)  & 0.0058  (0.0008) & 0.0051  (0.0036) & $\boldsymbol{0.0035}$  (0.0176) \\
                           & $\beta_2$ & 0.0295  (0.0042) & 0.0215  (-0.0001) & 0.0190  (-0.0155) & $\boldsymbol{0.0118}$   (-0.0125) \\
		\noalign{\smallskip}\hline\noalign{\smallskip}
	\end{tabular}
\end{table}

\begin{table}[p]
\caption{MSE and bias when $n=200$ with $0.6$ to $0.8$ censoring rate (Boldfaced numbers indicate the smallest value in each method)}
\label{tab:n200,c60} 
	\centering
	\begin{tabular}{ c  c  c   c    c   c  } 
        \hline \hline \noalign{\smallskip}
	\multirow{2}{*}{Method}	&
	\multirow{2}{*}{}	&
	$\text{N}(0,1)$ & $\text{t}(3)$ & $\text{Gumbel}(0,5)$ & $\text{skew-t}(0,1,-15,3)$ \\
 & & MSE  (bias) & MSE  (bias) & MSE  (bias) & MSE  (bias) \\ \hline
    \multirow{3}{*}{Normal} & $\beta_0$ & $\boldsymbol{0.0482}$ (0.1146) & 0.4601 (0.5456) & 0.0823 (0.1682) & 0.8246 (0.7750) \\
                            & $\beta_1$ &  0.0169 (0.0449) & 0.0847 (0.1931) & 0.0270 (0.1034) & 0.1153  (0.2148) \\
                            & $\beta_2$ & 0.0463 (-0.0437) & 0.1475 (-0.2346) & 0.0699 (-0.1165) & 0.2208 (-0.2734) \\
\hline
    \multirow{3}{*}{SN} & $\beta_0$ & 0.0574  (0.1470) & 0.3467 (0.4781) & $\boldsymbol{0.0468}$ (0.1323) & 0.2299 (0.3527) \\
                            & $\beta_1$ & 0.0162  (0.0291) & 0.0727 (0.1716) & 0.0103 (0.0080) & 0.1037 (0.2606) \\
                            & $\beta_2$ & 0.0447 (-0.0253) & 0.1345 (-0.2064) & $\boldsymbol{0.0294}$ (-0.0109) & 0.1863 (-0.3075) \\
\hline
    \multirow{3}{*}{Gehan} & $\beta_0$ &  0.0690 (0.1711) & 0.1760 (0.3020) & 0.1354 (0.2769) & 0.1347 (0.1150) \\
                           & $\beta_1$ &  $\boldsymbol{0.0143}$ (0.0105) & $\boldsymbol{0.0197}$ (0.0068) & $\boldsymbol{0.0096}$ (0.0096) & 0.0452 (0.0270) \\
                           & $\beta_2$ & $\boldsymbol{0.0418}$ (-0.0069) & $\boldsymbol{0.0509}$ (-0.0277) & 0.0311 (-0.0120) & 0.1182 (-0.0798) \\
\hline
    \multirow{3}{*}{GEE} & $\beta_0$ &  0.0753 (0.1931) & 0.2668 (0.3800) & 0.1593 (0.3105) & 0.2380 (0.1665) \\
                          & $\beta_1$ & 0.0165 (0.0394) & 0.0470 (0.0847) & 0.0193 (0.0606) & 0.0733 (0.0399) \\
                          & $\beta_2$ & 0.0449 (-0.0340) & 0.0949  (-0.1163) & 0.0543  (-0.0619) & 0.1839   (-0.1132) \\
          \hline
    \multirow{3}{*}{SSNSM} & $\beta_0$ &  0.0645 (0.1286) & $\boldsymbol{0.1098}$  (0.1790) & 0.0632  (0.1511) & $\boldsymbol{0.0613}$ (0.1332) \\
                           & $\beta_1$ &  0.0188 (0.0139) & 0.0248 (0.0259) & 0.0105 (-0.0000) & $\boldsymbol{0.0194}$ (0.0807)\\
                           & $\beta_2$ & 0.0519 (-0.0216) & 0.0612 (-0.0533) & 0.0310 (0.0016) & $\boldsymbol{0.0545}$  (-0.0704) \\
		\noalign{\smallskip}\hline\noalign{\smallskip}
	\end{tabular}
\end{table}

\begin{table}[p]
\caption{MSE and bias when $n=400$ with $0.3$ to $0.5$ censoring rate (Boldfaced numbers indicate the smallest value in each method)}
\label{tab:n400,c30} 
	\centering
	\begin{tabular}{ c  c  c   c    c   c  } 
        \hline \hline \noalign{\smallskip}
	\multirow{2}{*}{Method}	&
	\multirow{2}{*}{}	&
	$\text{N}(0,1)$ & $\text{t}(3)$ & $\text{Gumbel}(0,5)$ & $\text{skew-t}(0,1,-15,3)$ \\
 & & MSE  (Bias) & MSE  (Bias) & MSE  (Bias) & MSE  (Bias) \\ \hline
    \multirow{3}{*}{Normal} & $\beta_0$ &  $\boldsymbol{0.0083}$ (0.0181) & 0.0162 (0.0802) & 0.0043 (-0.0127) & 0.0382 (0.1653) \\
                            & $\beta_1$ &  $\boldsymbol{0.0039}$ (0.0048) & 0.0066 (0.0363) & 0.0031 (0.0097) & 0.0090 (0.0515) \\
                            & $\beta_2$ & 0.0133 (-0.0065) & 0.0157 (-0.0294) & 0.0108 (-0.0045) & 0.0257 (-0.0564) \\
\hline
    \multirow{3}{*}{SN} & $\beta_0$ & 0.0086  (0.0227) & 0.0140 (0.0555) & $\boldsymbol{0.0037}$ (0.0173) & 0.0272 (-0.1464) \\
                            & $\beta_1$ & 0.0038  (0.0038) & 0.0053 (0.0284) & $\boldsymbol{0.0021}$ (-0.0010) & 0.0027 (0.0351) \\
                            & $\beta_2$ & 0.0134 (-0.0071) & 0.0130 (-0.0219) & $\boldsymbol{0.0071}$ (0.0034) & 0.0068 (-0.0379) \\
\hline
    \multirow{3}{*}{Gehan} & $\beta_0$ &  0.0087 (0.0230) & 0.0118 (0.0541) & 0.0068 (0.0429) & 0.0041 (0.0033) \\
                           & $\beta_1$ &  0.0040 (0.0016) & $\boldsymbol{0.0026}$ (0.0003) & 0.0024 (0.0015) & 0.0027 (0.0010) \\
                           & $\beta_2$ & $\boldsymbol{0.0133}$ (-0.0032) & $\boldsymbol{0.0086}$ (-0.0007) & 0.0084 (0.0059) & 0.0096  (0.0007) \\
\hline
    \multirow{3}{*}{GEE} & $\beta_0$ &  0.0088 (0.0251) & 0.0149 (0.0608) & 0.0084 (0.0503) & 0.0097 (0.0088) \\
                         & $\beta_1$ &  0.0039 (0.0053) & 0.0050 (0.0151) & 0.0034 (0.0131) & 0.0058 (0.0003)\\
                         & $\beta_2$ & 0.0134 (-0.0050) & 0.0142 (-0.0084) & 0.0119 (-0.0071) & 0.0240  (-0.0066) \\                          
          \hline
    \multirow{3}{*}{SSNSM} & $\beta_0$ &  0.0088 (0.0173) & $\boldsymbol{0.0074}$ (0.0138) & 0.0041 (0.0151) & $\boldsymbol{0.0028}$ (0.0217) \\
                           & $\beta_1$ &  0.0043 (0.0008)  & 0.0028 (0.0020) & 0.0023 (-0.0039) & $\boldsymbol{0.0013}$ (0.0090)\\
                           & $\beta_2$ & 0.0142 (-0.0050) & 0.0089 (-0.0021)  & 0.0073 (0.0068) & $\boldsymbol{0.0046}$ (-0.0163) \\
		\noalign{\smallskip}\hline\noalign{\smallskip}
	\end{tabular}
\end{table}

\begin{table}[p]
\caption{MSE and bias when $n=400$ with $0.6$ to $0.8$ censoring rate (Boldfaced numbers indicate the smallest value in each method)}
\label{tab:n400,c60} 
	\centering
	\begin{tabular}{ c  c  c   c    c   c  } 
        \hline \hline \noalign{\smallskip}
	\multirow{2}{*}{Method}	&
	\multirow{2}{*}{}	&
	$\text{N}(0,1)$ & $\text{t}(3)$ & $\text{Gumbel}(0,5)$ & $\text{skew-t}(0,1,-15,3)$ \\
 & & MSE  (Bias) & MSE  (Bias) & MSE  (Bias) & MSE  (Bias) \\ \hline
    \multirow{3}{*}{Normal} & $\beta_0$ &  $\boldsymbol{0.0297}$ (0.0673) & 0.3596 (0.5191) & 0.0263 (0.0753) & 0.7383 (0.7838) \\
                            & $\beta_1$ &  0.0091 (0.0273) & 0.0621 (0.2021) & 0.0130 (0.0706) & 0.0894  (0.2310) \\
                            & $\beta_2$ & 0.0230 (-0.0194) & 0.0943 (-0.2158) & 0.0264 (-0.0767) & 0.1233 (-0.2158) \\
\hline
    \multirow{3}{*}{SN} & $\beta_0$ &  0.0348 (0.0965) & 0.2712 (0.4601) & $\boldsymbol{0.0183}$ (0.0762) & 0.1663 (0.3416) \\
                            & $\beta_1$ & 0.0085  (0.0179) & 0.0561 (0.1882) & 0.0042 (0.0038) & 0.0892 (0.2705) \\
                            & $\beta_2$ & 0.0220 (-0.0114) & 0.0874 (-0.2050) & $\boldsymbol{0.0122}$ (-0.0054) & 0.1183 (-0.2674) \\
\hline
    \multirow{3}{*}{Gehan} & $\beta_0$ &  0.0447 (0.1236) & 0.1493 (0.3015) & 0.0942 (0.2373) & 0.0423 (0.0288) \\
                           & $\beta_1$ &  $\boldsymbol{0.0080}$ (0.0053) & 0.0104 (0.0184) & $\boldsymbol{0.0042}$ (0.0026) & 0.0182 (0.0083) \\
                           & $\beta_2$ & $\boldsymbol{0.0219}$ (0.0017) & 0.0220 (-0.0194) & 0.0129 (0.0004) & 0.0417 (0.0024) \\
\hline
    \multirow{3}{*}{GEE} & $\beta_0$ & 0.0483 (0.1354) & 0.1825 (0.3444) & 0.1058 (0.2624) & 0.0746 (0.0342) \\
                         & $\beta_1$ &  0.0089 (0.0238) & 0.0208 (0.0728) & 0.0090 (0.0442) & 0.0302 (0.0093) \\
                         & $\beta_2$ & 0.0230 (-0.0123) & 0.0415 (-0.0763) & 0.0211 (-0.0374) & 0.0742  (0.0097) \\
          \hline
    \multirow{3}{*}{SSNSM} & $\beta_0$ &  0.0337 (0.0703) & $\boldsymbol{0.0346}$ (0.0948) & 0.0211 (0.0703) & $\boldsymbol{0.0281}$ (0.1101) \\
                           & $\beta_1$ &  0.0090 (0.0095)  & $\boldsymbol{0.0093}$ (0.0251) & 0.0044 (-0.0071) & $\boldsymbol{0.0110}$ (0.0689) \\
                           & $\beta_2$ & 0.0245 (-0.0015) & $\boldsymbol{0.0217}$ (-0.0324) & 0.0127 (0.0063) & $\boldsymbol{0.0233}$ (-0.0515) \\
		\noalign{\smallskip}\hline\noalign{\smallskip}
	\end{tabular}
\end{table}

The results, including MSE and bias for each method, are presented in Tables \ref{tab:n200,c30} - \ref{tab:n400,c60}.
From Tables \ref{tab:n200,c30} - \ref{tab:n400,c60}, when the normally distributed error is considered,
all estimation methods seem to provide similar performance, regardless of the censoring rates or sample sizes. However, the performance of the Normal estimator get worse when the error distributions deviate from the normality due to the violation of parametric assumptions.
In Tables \ref{tab:n200,c30} - \ref{tab:n400,c30}, when the true error distribution is $t(3)$, the Gehan estimators perform the best in estimating regression coefficients, with the exception of the intercept term. Interestingly, the SSNSM method demonstrates comparable performance to Gehan and emerges as the preferred choice for estimating the intercept term. However, in Table \ref{tab:n400,c60}, SSNSM performs the best for all coefficients.

Under Gumbel error distributions, SN exhibits the most accurate estimation for nearly all regression coefficients, closely followed by SSNSM. The efficiency of the SSNSM estimator remains high even under a Gumbel error distribution, which does not fall within the class of skew-normal scale mixtures. 
When the error distribution is $\text{skew-t}(0,1,-15,3)$, the SSNSM estimator demonstrates superior performance compared to other estimation methods. In general, the SSNSM yields results that are close to the best when the true error distribution is normal, $\text{t}$ and Gumbel, while it performs the best when the true error distribution is $\text{skew-t}$. The SSNSM method showcases overall superiority in handling a wide range of error distributions, exhibiting robustness and efficiency in parameter estimation, particularly in cases involving heavy-tailed distributions regardless of skewness.

\subsection{Simulation 2}   
\label{sec5-2}

In this subsection, we investigate the accuracy of failure time predictions for each method within the contexts outlined in Section \ref{sec5-1}. 
For the training dataset, we consider a sample size of 250, while for the test dataset, the sample size is $m = 1000$. To calculate the predicted failure times, we employ the covariate information from the test dataset and the regression coefficient estimates derived from the training dataset. Specifically, we compute the root mean squared error of predictions ($RMSEP$) as 
\begin{align*}
\text{RMSEP}(j) = \sqrt{ \frac{1}{m} \sum_{i = 1}^{m} (\log T_{ij}  - \bm{z}_{ij}^{\top} \hat{{\betab}}_{j} )^2}, 
\end{align*}
where $T_{ij}$ represents the true failure time for the $i$th observation in the $j$th test dataset, $\bm{z}_{ij}^{\top}$ denotes $(1, \xb_{ij}^{\top})^{\top}$, where $\xb_{ij}$ is the $i$th covariate vector of the $j$th test dataset, and $\hat{\betab}_j$ is the vector of estimated regression coefficients from the $j$th training dataset.
The boxplots of $\text{RMSEP}(j)$, $j=1,2,\ldots,200$, are presented in Figures 1 and 2 based on 200 replications.

\begin{figure}[p] 
\centering
\subfigure[$\text{N}(0,1)$]{
\includegraphics[width=0.475\linewidth]{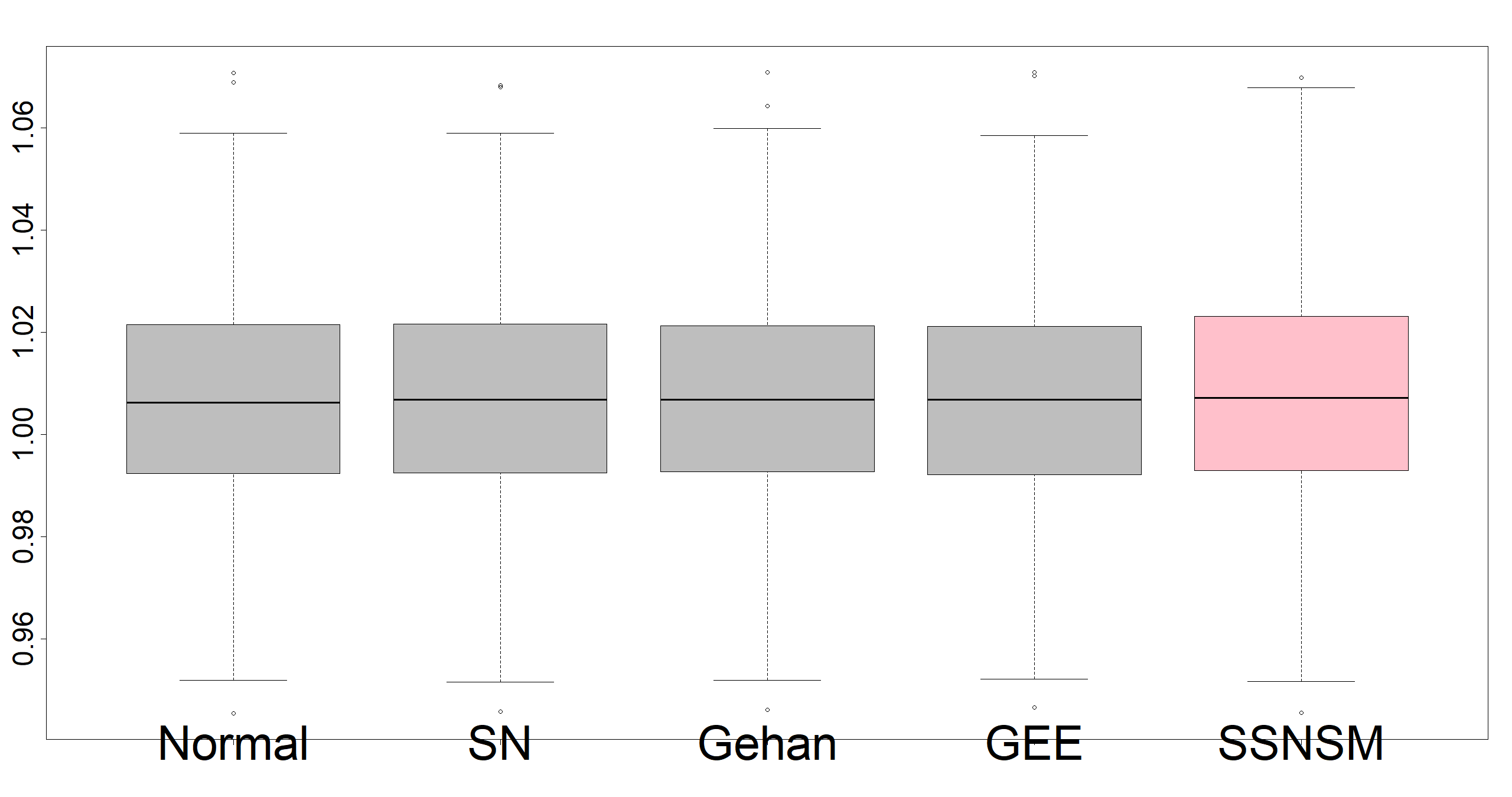}
}
\centering
\subfigure[$\text{t}(3)$]{
\includegraphics[width=0.475\linewidth]{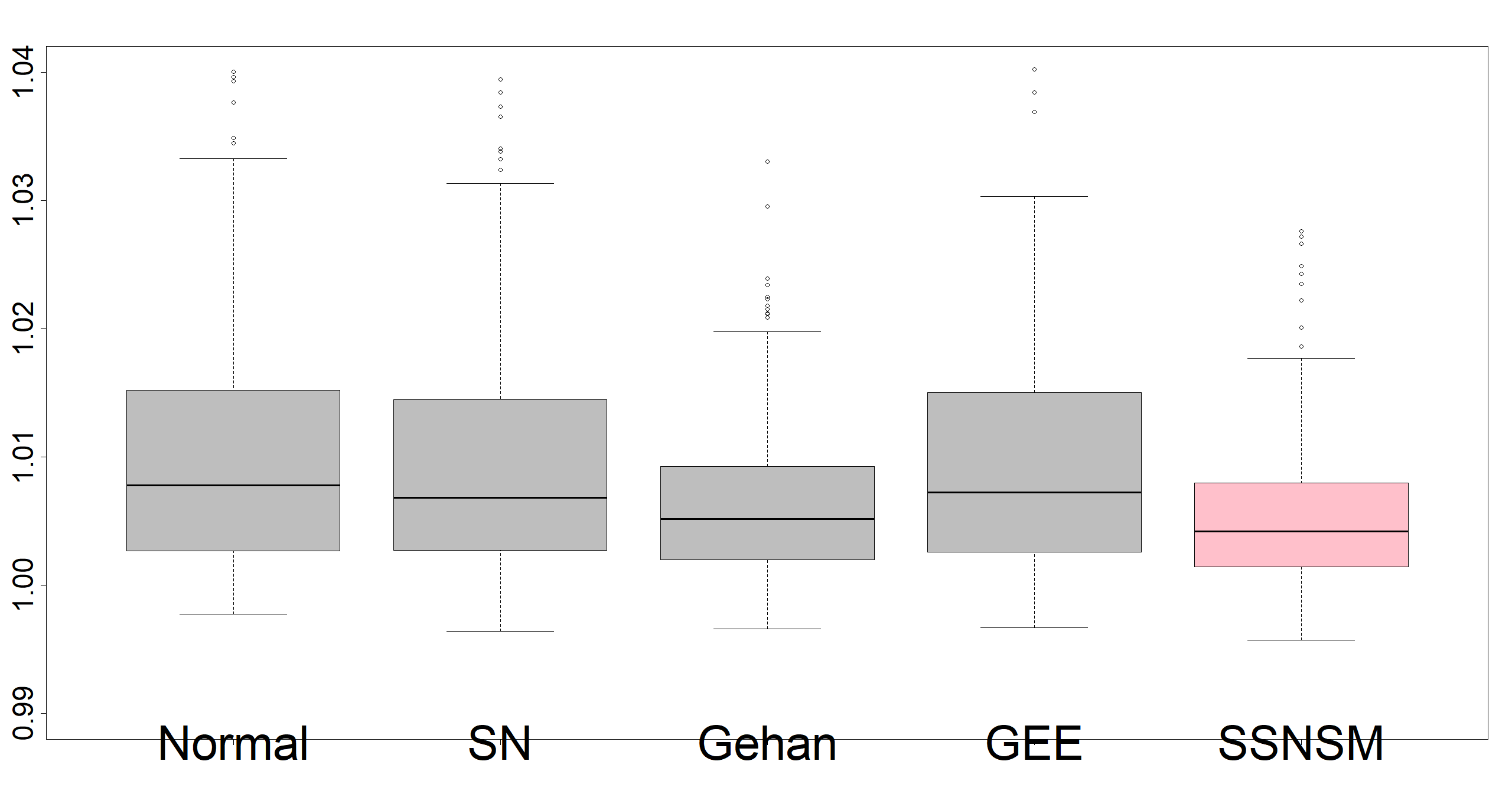}
}
\centering
\subfigure[$\text{Gumbel}(0,5)$]{
\includegraphics[width=0.475\linewidth]{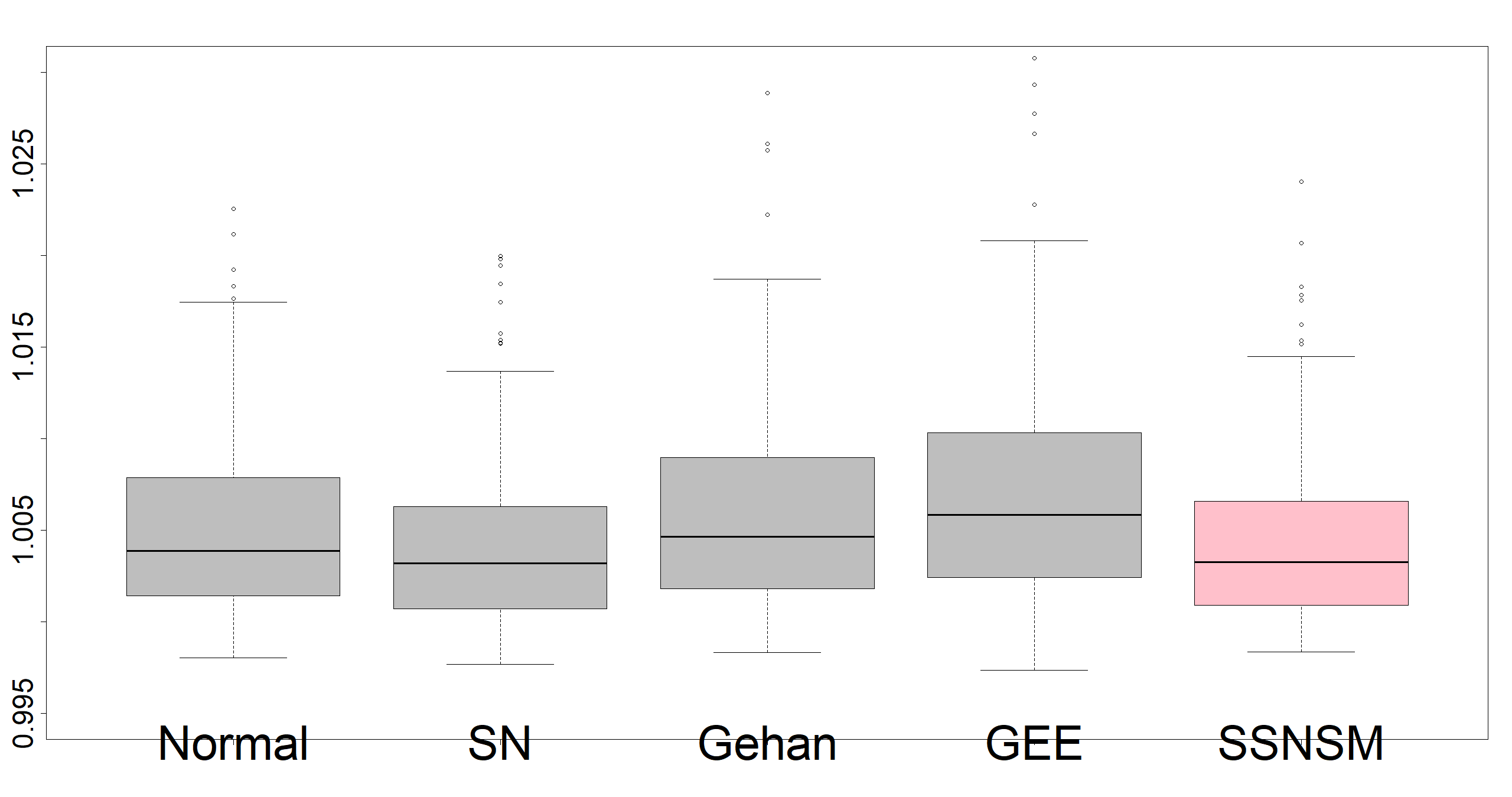}
}
\centering
\subfigure[$\text{skew-t}(0,1,-15,3)$]{
\includegraphics[width=0.475\linewidth]{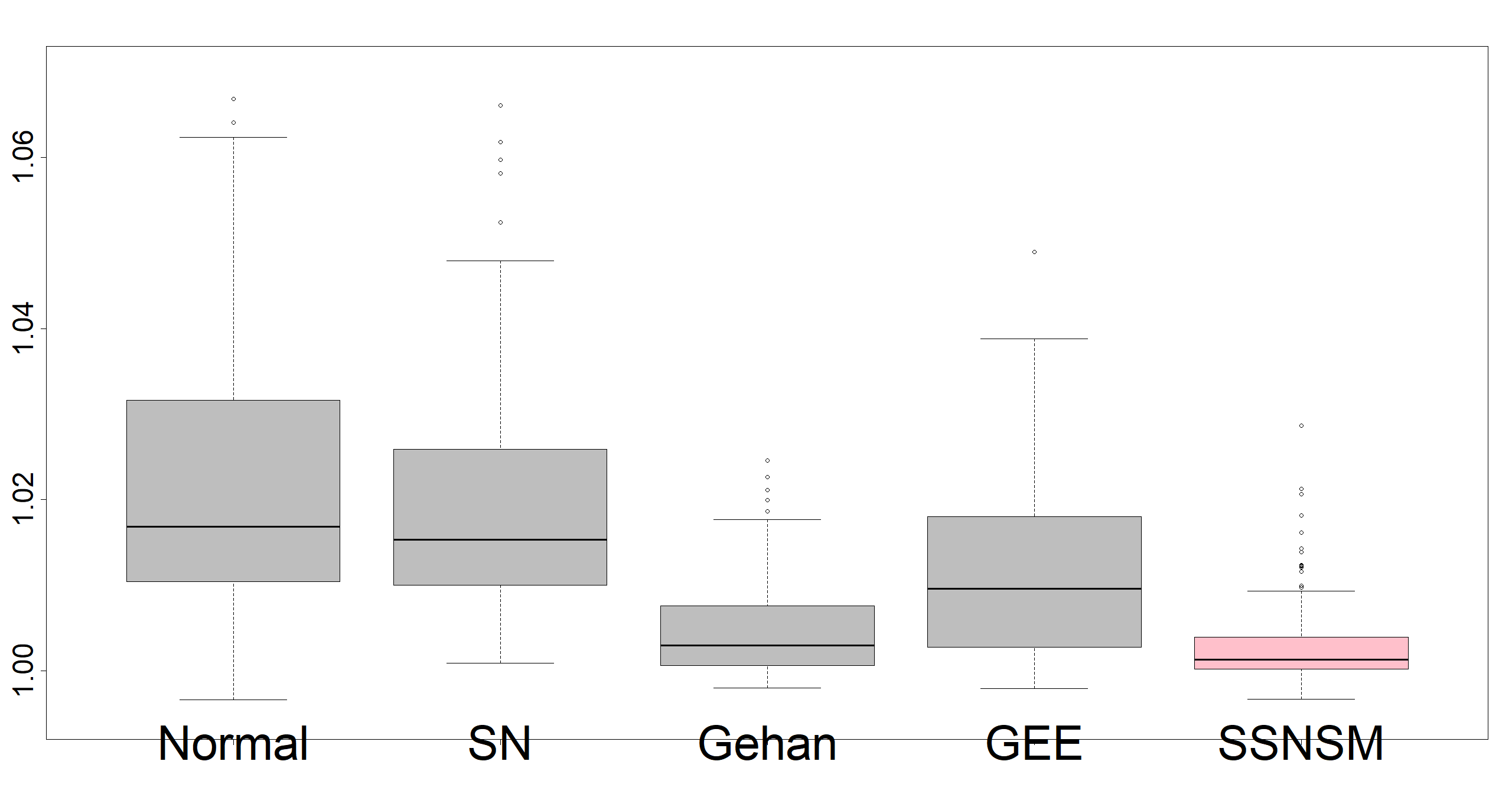}
}
\label{prdiction30}
\caption{Comparison of prediction performance with box plots of $\text{RMSEP}$ in each case when censoring rate is 0.3 to 0.5}
\end{figure}

\begin{figure}[p] 
\label{prdiction60}
\centering
\subfigure[$\text{N}(0,1)$]{
\includegraphics[width=0.475\linewidth]{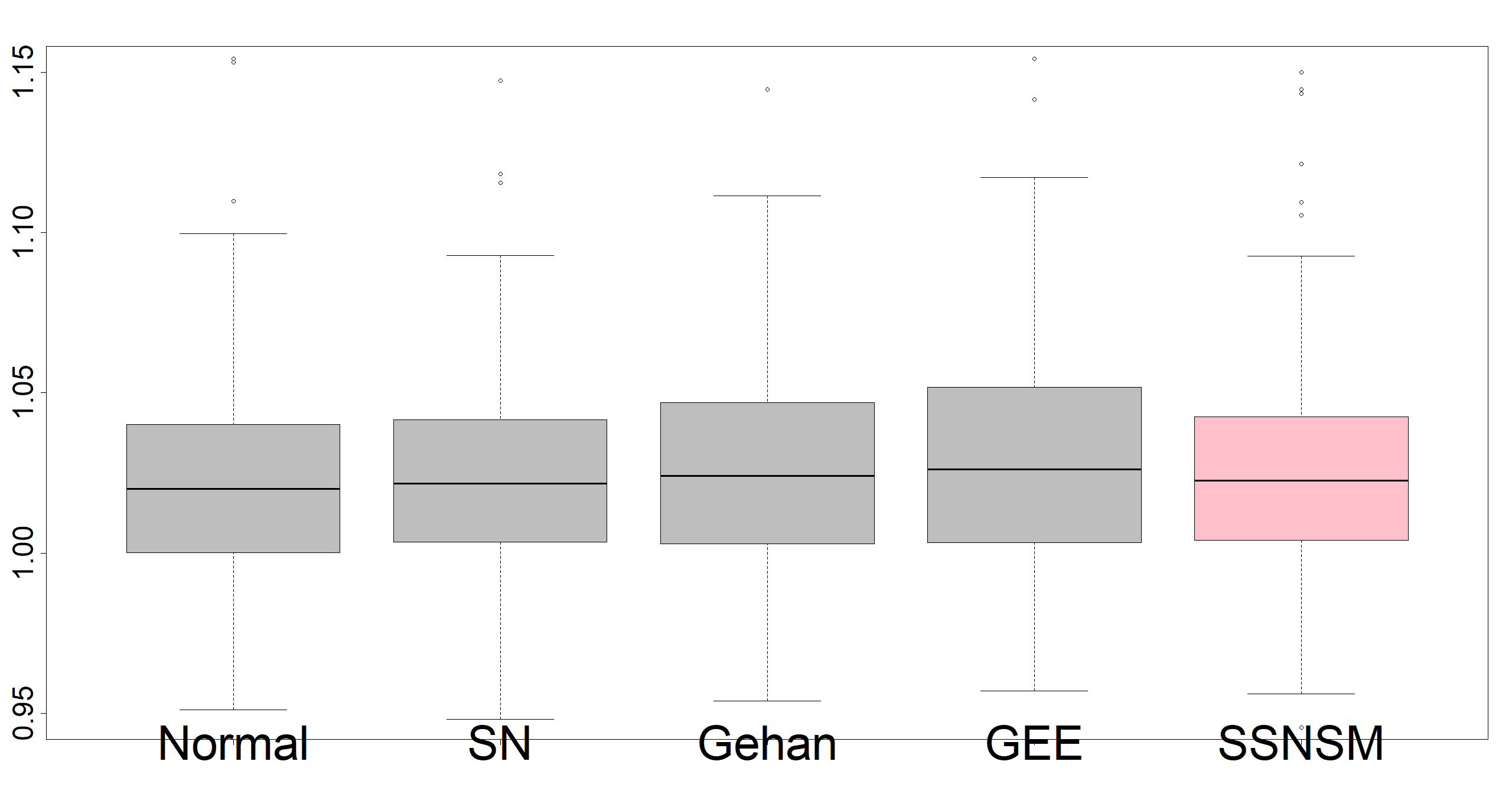}
}
\centering
\subfigure[$\text{t}(3)$]{
\includegraphics[width=0.475\linewidth]{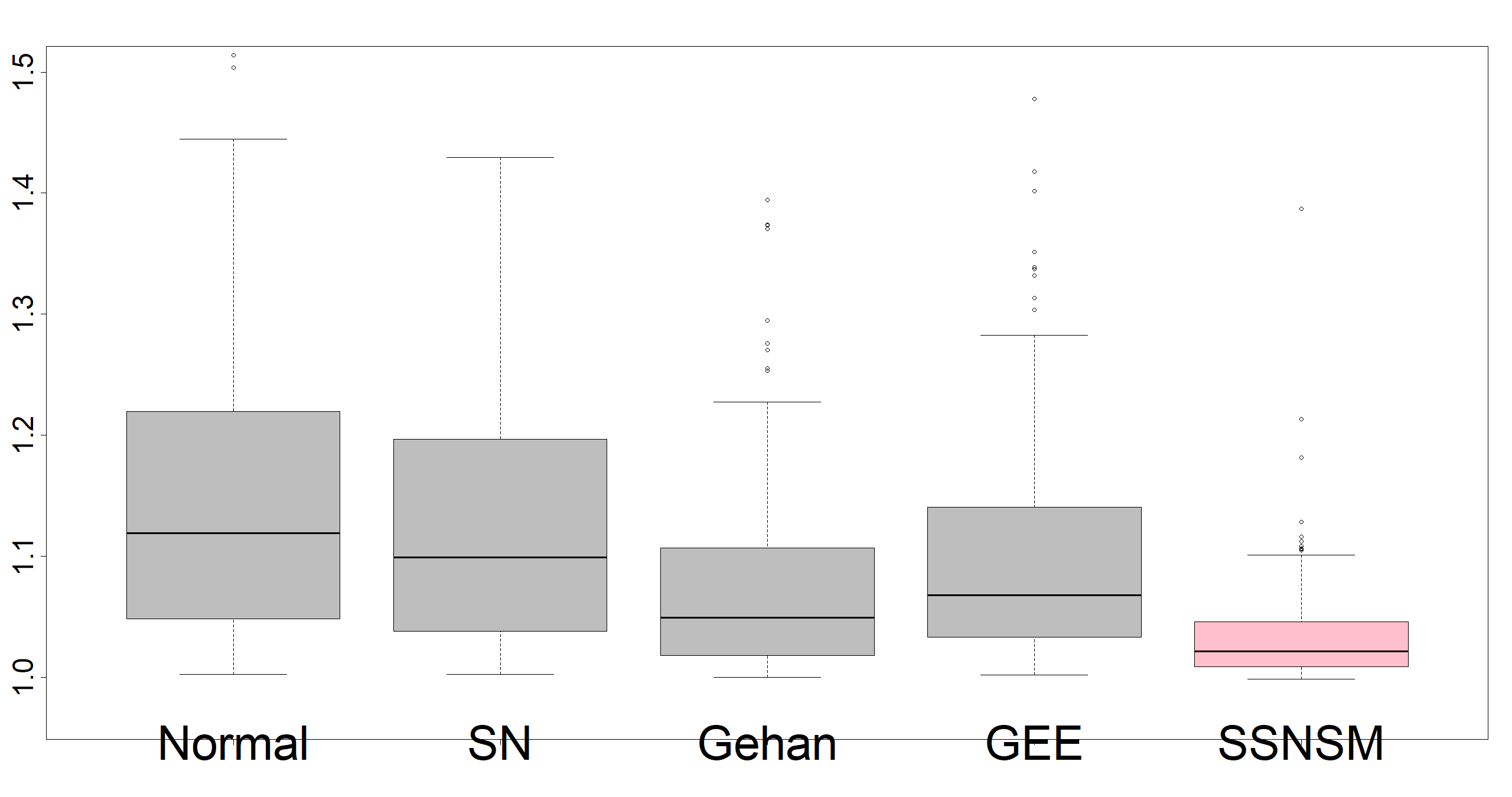}
}
\centering
\subfigure[$\text{Gumbel}(0,5)$]{
\includegraphics[width=0.475\linewidth]{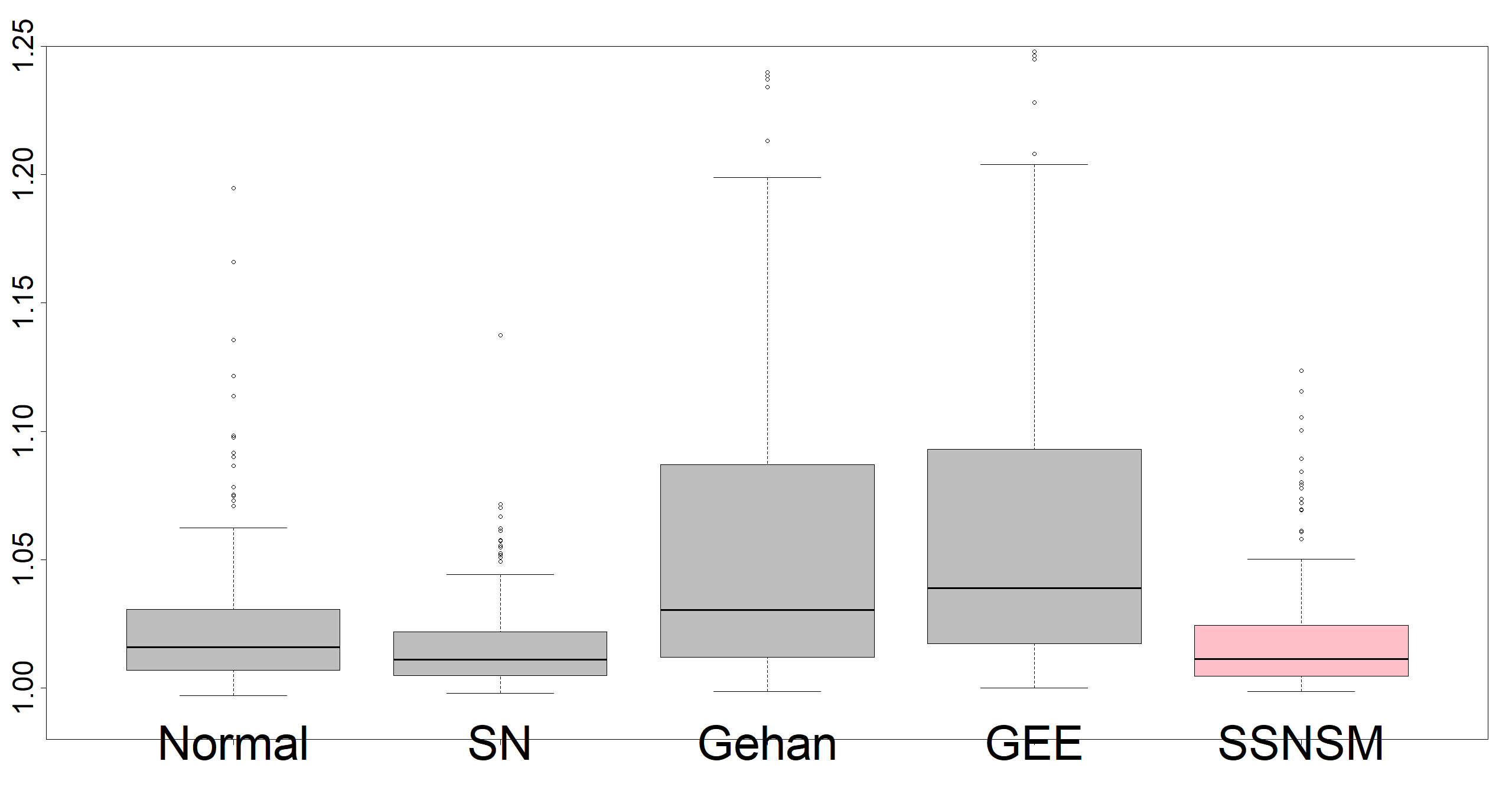}
}
\centering
\subfigure[$\text{skew-t}(0,1,-15,3)$]{
\includegraphics[width=0.475\linewidth]{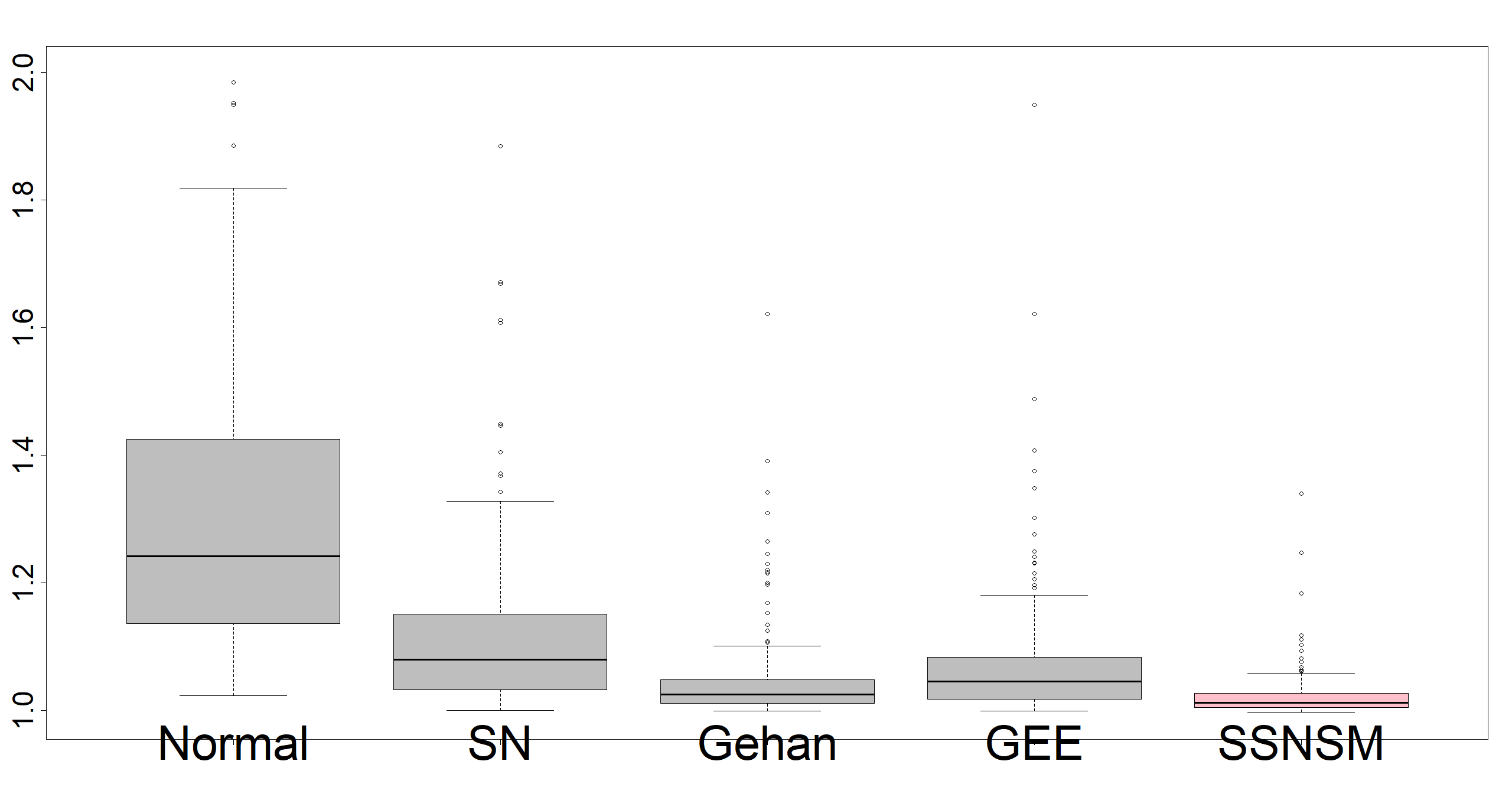}
}
\caption{Comparison of prediction performance with box plots of $\text{RMSEP}$ in each case when censoring rate is 0.6 to 0.8}
\end{figure}

When errors follow to the normal distribution, the Normal method exhibits superior performance, while the SSNSM method exhibits slightly lower performance in terms of the interquartile range when the censoring rate is between 0.3 and 0.5. However, it ranks as the second-best performer when the censoring rate ranges from 0.6 to 0.8.
For errors following to a t-distribution, SSNSM outperforms other methods, closely followed by Gehan. 

For errors following the Gumbel distribution, SN can be considered as the most effective estimation method, closely followed by SSNSM. In the case of errors following the skew-t distribution, SSNSM outperforms other methods, with Gehan being the closest competitor. Note that Gehan and GEE exhibit significant performance degradation as the censoring rate increases when errors follow the Gumbel distribution. Furthermore, the Normal method experiences a substantial decline in performance when errors follow the heavy tailed distributions. 
On the other hand, the SSNSM performs the best in almost all cases. 
These observed patterns in prediction outcomes seem to be linked to the estimated intercept term $\beta_0$ in Section \ref{sec5-1}.

\subsection{Simulation 3}
\label{sec5-3}

In the third simulation study, we aim to examine the performance of different methods in handling data with increasing levels of skewness and the presence of outliers. To achieve this, we generate data using a skew-t distribution with $3$ degrees of freedom and vary the slant parameters across different scenarios ($\lambda = -1, -4, -10, -50$) as shown in Figure 3. By manipulating the slant parameters, we can control the level of skewness in the data, with higher values indicating greater skewness. The sample size for this study is set to $n=200$ and $n=400$, while the remaining settings remain consistent with those of the first simulation study.

\begin{figure}[ht] 
\centering
\includegraphics[width=0.45\linewidth]{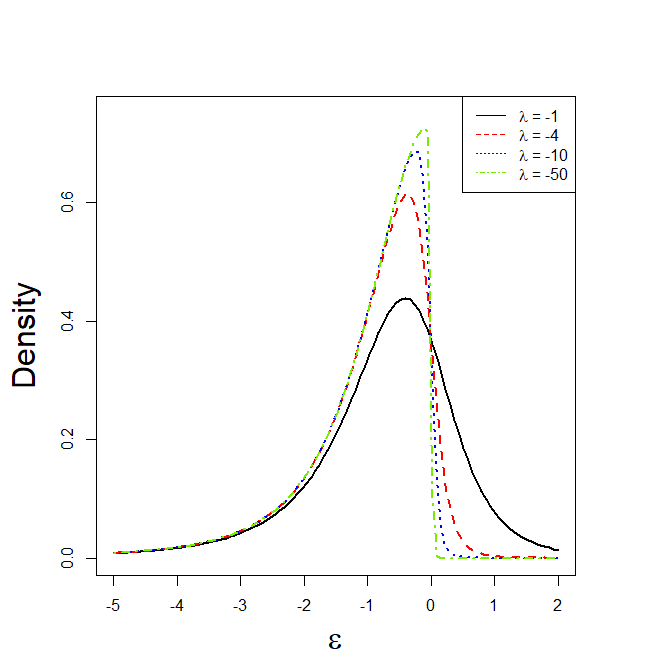}
\caption{Density of $\epsilon$ from skew-t distribution across various $\lambda$}
\label{true_error}
\end{figure}

\begin{table}[p]
\caption{Estimates of regression coefficients when $n=200$ with $0.3$ to $0.5$ censoring rate (Boldfaced numbers indicate the smallest value in each criterion)}
\label{tab:skewness1} 
	\centering
	\begin{tabular}{ c  c  c   c    c   c  } 
        \hline \hline \noalign{\smallskip}
	\multirow{2}{*}{Method}	&
	\multirow{2}{*}{}	&
	$\lambda = -1$ & $\lambda = -4$ & $\lambda = -10$ & $\lambda = -50$\\
 & & MSE  (Bias) & MSE  (Bias) & MSE  (Bias) & MSE  (Bias) \\ \hline
    \multirow{3}{*}{Normal} & $\beta_0$ &  0.0322  (0.1293) & 0.0394 (0.1474) & 0.0451  (0.1662) & 0.0460  (0.1688)  \\
                            & $\beta_1$ &  0.0124  (0.0431) & 0.0150 (0.0476) & 0.0148  (0.0479) & 0.0151  (0.0499)  \\
                            & $\beta_2$ & 0.0356  (-0.0544) & 0.0419 (-0.0505)& 0.0441  (-0.0773) & 0.0443  (-0.0767)  \\
\hline
    \multirow{3}{*}{SN} & $\beta_0$ & 0.0165  (0.0523) & 0.0154 (-0.0375) &0.0231 (-0.1093) & 0.0289 (-0.1494)  \\
                            & $\beta_1$ & 0.0116  (0.0523) & 0.0110 (0.0648) & 0.0059 (0.0431) & 0.0031 (0.0347)  \\
                            & $\beta_2$ & 0.0298 (-0.0595) & 0.0286 (-0.0552) & 0.0174 (-0.0424) & 0.0079 (-0.0339)  \\
\hline
    \multirow{3}{*}{Gehan} & $\beta_0$ &  0.0101  (0.0424) & $\boldsymbol{0.0101}$ (0.0210)  & 0.0092  (0.0214) & 0.0091  (0.0201)  \\
                           & $\beta_1$ &  0.0059  (0.0063) & 0.0073 (0.0095) & 0.0061  (0.0065) & 0.0060  (0.0077)  \\
                           & $\beta_2$ & 0.0176  (-0.0143) &  0.0226 (-0.0162) & 0.0205 (-0.0232) & 0.0201  (-0.0213) \\
\hline
    \multirow{3}{*}{GEE} & $\beta_0$ &  0.0164  (0.0508) & 0.0167 (0.0206) & 0.0185  (0.0310) & 0.0187  (0.0311) \\
                         & $\beta_1$ &  0.0110  (0.0128) & 0.0129 (0.0065) & 0.0139 (0.0037) & 0.0142  (0.0048)  \\
                         & $\beta_2$ & 0.0333  (-0.0267) & 0.0410 (-0.0132) & 0.0424  (-0.0374) & 0.0420  (-0.0366)  \\
          \hline
    \multirow{3}{*}{SSNSM} & $\beta_0$ &  $\boldsymbol{0.0091}$ (0.0173) & 0.0113 (0.0334) & $\boldsymbol{0.0077}$  (0.0281) & $\boldsymbol{0.0027}$  (0.0147)  \\
                           & $\beta_1$ &  $\boldsymbol{0.0057}$  (0.0115) & $\boldsymbol{0.0071}$ (0.0324) & $\boldsymbol{0.0040}$  (0.0181) & $\boldsymbol{0.0012}$  (0.0110)  \\
                           & $\beta_2$ & $\boldsymbol{0.0150}$  (-0.0102) & $\boldsymbol{0.0192}$ (-0.0228) & $\boldsymbol{0.0155}$  (-0.0141) & $\boldsymbol{0.0039}$ (-0.0031)  \\
		\noalign{\smallskip}\hline\noalign{\smallskip}
	\end{tabular}
\end{table}

\begin{table}[p]
\caption{Estimates of regression coefficients when $n=200$ with $0.6$ to $0.8$ censoring rate (Boldfaced numbers indicate the smallest value in each criterion)}
\label{tab:skewness2} 
	\centering
	\begin{tabular}{ c  c  c   c    c   c  } 
        \hline \hline \noalign{\smallskip}
	\multirow{2}{*}{Method}	&
	\multirow{2}{*}{}	&
	$\lambda = -1$ & $\lambda = -4$ & $\lambda = -10$ & $\lambda = -50$\\
 & & MSE  (Bias) & MSE  (Bias) & MSE  (Bias) & MSE  (Bias) \\ \hline
    \multirow{3}{*}{Normal} & $\beta_0$ & 0.5353 (0.6353) & 0.6855 (0.7114) & 0.8181 (0.7727) & 0.8128 (0.7704) \\
                            & $\beta_1$ &  0.0843 (0.2124) & 0.1020 (0.2052) & 0.1166 (0.2141) & 0.1102 (0.2108)  \\
                            & $\beta_2$ & 0.1448 (-0.2194) & 0.1611 (-0.2110) & 0.2229 (-0.2755) & 0.2194 (-0.2731) \\
\hline
    \multirow{3}{*}{SN} & $\beta_0$ &  0.3707 (0.5275) & 0.2297 (0.3881) & 0.2424 (0.3657) & 0.2341 (0.3492) \\
                            & $\beta_1$ & 0.1078  (0.2560) &  0.1072 (0.2705) & 0.1042 (0.2607) & 0.1033 (0.2596)  \\
                            & $\beta_2$ & 0.1769 (-0.2806) &  0.1556 (-0.2662) & 0.1994 (-0.3165) & 0.1935 (-0.3140)  \\
\hline
    \multirow{3}{*}{Gehan} & $\beta_0$ &  $\boldsymbol{0.1265}$ (0.2122) & $\boldsymbol{0.1102}$ (0.1054)  & 0.1344 (0.1202) & 0.1438 (0.1057)  \\
                           & $\beta_1$ &  $\boldsymbol{0.0277}$ (0.0288) & $\boldsymbol{0.0400}$ (0.0291) & 0.0475 (0.0273) & 0.0439 (0.0251) \\
                           & $\beta_2$ & $\boldsymbol{0.0678}$ (-0.0371) & $\boldsymbol{0.1068}$ (-0.0373)  & 0.1162 (-0.0809) & 0.1216 (-0.0792)  \\
\hline
    \multirow{3}{*}{GEE} & $\beta_0$ &  0.1859 (0.2790)  & 0.1596 (0.1344) & 0.2411 (0.1718) & 0.2398 (0.1591)  \\
                          & $\beta_1$ & 0.0474 (0.0855)  & 0.0553 (0.0369)  & 0.0750 (0.0410) & 0.0719 (0.0360)  \\
                          & $\beta_2$ & 0.1089 (-0.0927) & 0.1342 (-0.0498) & 0.1862  (-0.1153) & 0.1852  (-0.1120)  \\
          \hline
    \multirow{3}{*}{SSNSM} & $\beta_0$ &  0.1876 (0.2742) &  0.1464 (0.2632) & $\boldsymbol{0.0824}$  (0.1723) & $\boldsymbol{0.0523}$  (0.0941) \\
                           & $\beta_1$ &  0.0545 (0.1060) & 0.0501 (0.1450) & $\boldsymbol{0.0252}$ (0.0942) & $\boldsymbol{0.0147}$ (0.0601) \\
                           & $\beta_2$ & 0.1048 (-0.1317) & 0.1098 (-0.1310) & $\boldsymbol{0.0748}$ (-0.1033) & $\boldsymbol{0.0491}$ (-0.0487)  \\
		\noalign{\smallskip}\hline\noalign{\smallskip}
	\end{tabular}
\end{table}

\begin{table}[p]
\caption{Estimates of regression coefficients when $n=400$ with $0.3$ to $0.5$ censoring rate (Boldfaced numbers indicate the smallest value in each criterion)}
\label{tab:skewness3} 
	\centering
	\begin{tabular}{ c  c  c   c    c   c  } 
        \hline \hline \noalign{\smallskip}
	\multirow{2}{*}{Method}	&
	\multirow{2}{*}{}	&
	$\lambda = -1$ & $\lambda = -4$ & $\lambda = -10$ & $\lambda = -50$\\
 & & MSE  (Bias) & MSE  (Bias) & MSE  (Bias) & MSE  (Bias) \\ \hline
    \multirow{3}{*}{Normal} & $\beta_0$ &  0.0245  (0.1263) & 0.0380 (0.1712) & 0.0383  (0.1655) & 0.0382  (0.1645)  \\
                            & $\beta_1$ &  0.0061  (0.0424) & 0.0093 (0.0576) & 0.0089  (0.0527) & 0.0090  (0.0510)  \\
                            & $\beta_2$ & 0.0215  (-0.0530) & 0.0253 (-0.0723)& 0.0260  (-0.0576) & 0.0255  (-0.0552)  \\
\hline
    \multirow{3}{*}{SN} & $\beta_0$ & 0.0087  (0.0316) & 0.0123 (-0.0625) &0.0212 (-0.1212) & 0.0367 (-0.1807)  \\
                            & $\beta_1$ & 0.0062  (0.0449) & 0.0054 (0.0484) & 0.0037 (0.0402) & 0.0015 (0.0257)  \\
                            & $\beta_2$ & 0.0169 (-0.0493) & 0.0131 (-0.0535) & 0.0098 (-0.0486) & 0.0030 (-0.0243)  \\
\hline
    \multirow{3}{*}{Gehan} & $\beta_0$ &  0.0050  (0.0267) & 0.0039 (0.0182)  & 0.0042  (0.0042) & 0.0041  (0.0019)  \\
                           & $\beta_1$ &  0.0024  (0.0030) & 0.0030 (0.0092) & 0.0027  (0.0018) & 0.0026  (0.0008)  \\
                           & $\beta_2$ & 0.0101  (-0.0043) &  0.0097 (-0.0175) & 0.0099 (0.0003) & 0.0095  (0.0015) \\
\hline
    \multirow{3}{*}{GEE} & $\beta_0$ &  0.0084  (0.0371) & 0.0088 (0.0227) & 0.0096  (0.0101) & 0.0100  (0.0072) \\
                         & $\beta_1$ &  0.0044  (0.0090) & 0.0061 (0.0080) & 0.0056 (0.0016) & 0.0059  (-0.0003)  \\
                         & $\beta_2$ & 0.0181  (-0.0198) & 0.0223 (-0.0237) & 0.0240  (-0.0076) & 0.0242  (-0.0053)  \\
          \hline
    \multirow{3}{*}{SSNSM} & $\beta_0$ &  $\boldsymbol{0.0036}$ (-0.0006) & $\boldsymbol{0.0034}$ (0.0064) & $\boldsymbol{0.0034}$  (0.0190) & $\boldsymbol{0.0010}$  (0.0123)  \\
                           & $\beta_1$ &  $\boldsymbol{0.0024}$  (0.0010) & $\boldsymbol{0.0023}$ (0.0074) & $\boldsymbol{0.0017}$  (0.0089) & $\boldsymbol{0.0004}$  (0.0060)  \\
                           & $\beta_2$ & $\boldsymbol{0.0082}$  (-0.0029) & $\boldsymbol{0.0055}$ (-0.0104) & $\boldsymbol{0.0052}$  (-0.0164) & $\boldsymbol{0.0014}$ (-0.0048)  \\
		\noalign{\smallskip}\hline\noalign{\smallskip}
	\end{tabular}
\end{table}

\begin{table}[p]
\caption{Estimates of regression coefficients when $n=400$ with $0.6$ to $0.8$ censoring rate (Boldfaced numbers indicate the smallest value in each criterion)}
\label{tab:skewness4} 
	\centering
	\begin{tabular}{ c  c  c   c    c   c  } 
        \hline \hline \noalign{\smallskip}
	\multirow{2}{*}{Method}	&
	\multirow{2}{*}{}	&
	$\lambda = -1$ & $\lambda = -4$ & $\lambda = -10$ & $\lambda = -50$\\
 & & MSE  (Bias) & MSE  (Bias) & MSE  (Bias) & MSE  (Bias) \\ \hline
    \multirow{3}{*}{Normal} & $\beta_0$ & 0.5862 (0.6850) & 0.6998 (0.7631) & 1.0067 (0.9035) & 1.0101 (0.9038) \\
                            & $\beta_1$ &  0.0764 (0.2274) & 0.0860 (0.2319) & 0.1227 (0.2759) & 0.1247 (0.2771)  \\
                            & $\beta_2$ & 0.1336 (-0.2567) & 0.1156 (-0.2149) & 0.2048 (-0.3049) & 0.2050 (-0.3007) \\
\hline
    \multirow{3}{*}{SN} & $\beta_0$ &  0.3743 (0.5503) & 0.2196 (0.3881) & 0.2787 (0.4394) & 0.2700 (0.4272) \\
                            & $\beta_1$ & 0.0990  (0.2732) &  0.1030 (0.2925) & 0.1199 (0.3089) & 0.1195 (0.3075)  \\
                            & $\beta_2$ & 0.1597 (-0.3050) &  0.1318 (-0.2839) & 0.1910 (-0.3483) & 0.1898 (-0.3459)  \\
\hline
    \multirow{3}{*}{Gehan} & $\beta_0$ &  $\boldsymbol{0.0747}$ (0.1912) & $\boldsymbol{0.0458}$ (0.0469)  & 0.0737 (0.0878) & 0.0741 (0.0762)  \\
                           & $\beta_1$ &  $\boldsymbol{0.0130}$ (0.0293) & $\boldsymbol{0.0176}$ (0.0115) & 0.0215 (0.0319) & 0.0212 (0.0307) \\
                           & $\beta_2$ & $\boldsymbol{0.0360}$ (-0.0474) & $\boldsymbol{0.0424}$ (-0.0009)  & 0.0593 (-0.0493) & 0.0605 (-0.0450)  \\
\hline
    \multirow{3}{*}{GEE} & $\beta_0$ &  0.1148 (0.2392)  & 0.0764 (0.0696) & 0.1202 (0.1156) & 0.1179 (0.1003)  \\
                          & $\beta_1$ & 0.0242 (0.0644)  & 0.0278 (0.0197)  & 0.0334 (0.0361) & 0.0333 (0.0322)  \\
                          & $\beta_2$ & 0.0614 (-0.0938) & 0.0697 (-0.0128) & 0.0997  (-0.0689) & 0.1002  (-0.0596)  \\
          \hline
    \multirow{3}{*}{SSNSM} & $\beta_0$ &  0.0819 (0.1547) &  0.1351 (0.2560) & $\boldsymbol{0.0423}$  (0.1573) & $\boldsymbol{0.0161}$  (0.0719) \\
                           & $\beta_1$ &  0.0241 (0.0658) & 0.0372 (0.1280) & $\boldsymbol{0.0160}$ (0.0869) & $\boldsymbol{0.0061}$ (0.0458) \\
                           & $\beta_2$ & 0.0430 (-0.0792) & 0.0857 (-0.1414) & $\boldsymbol{0.0279}$ (-0.0852) & $\boldsymbol{0.0127}$ (-0.0326)  \\
		\noalign{\smallskip}\hline\noalign{\smallskip}
	\end{tabular}
\end{table}

The performance of the methods under different skewness levels is presented in Tables \ref{tab:skewness1} - \ref{tab:skewness4}. Across scenarios with 0.3 to 0.5 censoring rates, the SSNSM method generally outperforms other approaches. However, for censoring rates of 0.6 to 0.8, Gehan shows superior performance specifically when $\lambda$ is set to $-1$ and $-4$. Nevertheless, SSNSM remains competitive, exhibiting the best performance among MLE methods in these specific cases. The performance of the SSNSM method improves with higher slant parameters, showcasing its adaptability and capacity to offer more precise estimates under increased skewness.
In contrast, the Normal and SN methods exhibit decreased performance as slant parameters increase due to potential misspecification issues and reduced efficiency.

\section{Real data analysis}   
\label{sec6}

In this section, we employ two real datasets.  
For evaluating prediction performance of survival probability during a specific period, brier score \citep{graf1999assessment} can be used as
\begin{align*}
BS( t^*) = \frac{1}{n} \sum_{i=1}^{n} \Bigg \{   \frac{\hat{S}(\log t^* | \xb_i)^2}{\hat{G}(\log(Y_i))} I(\log(Y_i) < \log t^*, \delta_i = 1)  +  \frac{(1-\hat{S}(\log t^* | \xb_i))^2}{\hat{G}(\log t^*)} I(\log(Y_i) > \log t^*)  \Bigg \}, 
\end{align*}
where $t^*$ indicates the specific time point at which the brier score is computed, $\hat{S}(\cdot)$ denotes the survival function predicted by the method, and $\hat{G}(\cdot)$ is the Kaplan-Meier estimate of the survival function corresponding to censoring. 
As the brier score can be influenced by the choice of a single time point, we compute the integrated brier score (IBS) defined as
\begin{align*}
IBS(t_{\text{max}}) = \frac{1}{t_{\text{max}}} \int_{0}^{t_{\text{max}}} BS(t) dt.  
\end{align*}

\subsection{Lung cancer dataset}

The dataset utilized in the first analysis is sourced from some patients with advanced lung cancer, originally collected by the North Central Cancer Treatment Group  \citep{loprinzi1994prospective}. It is accessible through the $\texttt{survival}$ package in \texttt{R} \citep{therneau2019package}. This dataset comprises a total of 228 patients and incorporates various prognostic variables, including Age, Sex, ECOG (ECOG performance score as assessed by physicians), log(Phy) (logarithm of the Karnofsky performance score assessed by physicians), log(Pat) (logarithm of the Karnofsky performance score assessed by patients), log(Cal) (logarithm of Calories consumed during meals), and Loss (Weight loss in pounds during the last six months).
Observations with missing covariate information were excluded, resulting in a dataset of 167 observations with a censoring rate of $28.1\%$.

The estimates obtained from the five methods are presented in Table \ref{tab:lung}. Standard error estimates are computed using 500 bootstrap replicates for all estimators. Generally, the estimates provided by the five methods are similar, except for the $\log$(Cal) estimate from SSNSM. In the case of SSNSM, the $\log$(Cal) estimate is nearly zero with a negative sign, while the estimates from the other methods are positive. Since $\log$(Cal) does not appear to statistically significant to the failure time in any of the methods presented in Table \ref{tab:lung}, it seems reasonable for SSNSM to produce a nearly zero estimate for this variable.

According to the results obtained from SSNSM, the effects of Sex and ECOG are statistically significant. In Figure \ref{survival_plot (lung)}, the conditional survival probability for each case is displayed using SSNSM. Estimated conditional survival curves are constructed for each Sex and ECOG while fixing the values of the other covariates.
Specifically, females tend to have a longer failure time than males, while an increase in the ECOG rating corresponds to a decrease in the failure time. 
Furthermore, the IBS obtained for Normal, SN, Gehan, GEE and SSNSM are $0.0904$, $0.0863$, $0.0932$, $0.0877$ and $0.0851$, respectively.
That is, SSNSM exhibits the best performance, closely followed by SN and GEE.
Note that we designate $t_{\text{max}}$ as $1826.25$, which corresponds to a span of 5 years.

\begin{table}[p]
\caption{Estimated regression coefficients and standard errors in parentheses from the Lung cancer dataset}
\label{tab:lung} 
	\centering
	\begin{tabular}{   c  c c  r r   c  } 
        \hline \hline \noalign{\smallskip}
	
	Method	&
	Intercept & Age & Sex & ECOG  \\ \hline
     \multirow{2}{*}{Normal} &  6.2178 & -0.0126 & 0.5132 & -0.4844   \\
       &  (5.2979) & (0.0098) & (0.1806) & (0.2336)  \\
     \multirow{2}{*}{SN} & 9.8074  & -0.0110 & 0.3875 & -0.5338   \\
       &  (4.4152) & (0.0100) & (0.1711) & (0.1854)  \\
       \multirow{2}{*}{Gehan} &  5.4936 & -0.0138 & 0.6585 & -0.5831   \\
       &  (6.0952) & (0.0120) & (0.2008) & (0.2530)  \\
          \multirow{2}{*}{GEE} & 6.0252 & -0.0119 & 0.4977 & -0.4633    \\
       &  (5.0622) & (0.0096) & (0.1812) & (0.2313)  \\
         \multirow{2}{*}{SSNSM} & 8.9619 & -0.0050 & 0.3820 & -0.5368  \\
      &  (5.5653) & (0.0114) & (0.1756) & (0.2047)   \\
               
        \hline \hline \noalign{\smallskip}
		Method	&
	$\log$(Phy) & $\log$(Pat) & $\log$(Cal) & Loss \\ \hline
  \multirow{2}{*}{Normal} &  -0.8375 & 0.3933 & 0.2830 & 0.0093   \\
     &  (1.0297) & (0.5287) & (0.1980) & (0.0060)   \\
    \multirow{2}{*}{SN} &  -1.3550 & 0.5289 & 0.0033 & 0.0104   \\
       &  (0.8007) & (0.4436) & (0.1974) & (0.0066)  \\
      \multirow{2}{*}{Gehan} &  -0.6304 & 0.4209 & 0.2322 & 0.0081   \\
      &  (1.2227) & (0.5848) & (0.2375) & (0.0091)  \\
     \multirow{2}{*}{GEE} & -0.8055 & 0.3580 & 0.3019 & 0.0091  \\
      &  (0.9936) & (0.5111) & (0.1932) & (0.0057)  \\
     \multirow{2}{*}{SSNSM} & -1.2288 & 0.5614 & -0.0237 & 0.0090  \\
      &  (1.0608) & (0.4626) & (0.2159) & (0.0073)  \\
\noalign{\smallskip}\hline\noalign{\smallskip}
	\end{tabular}
\end{table}

\begin{figure}[p] 
\centering
\subfigure[Sex]{
\includegraphics[width=0.42\linewidth]{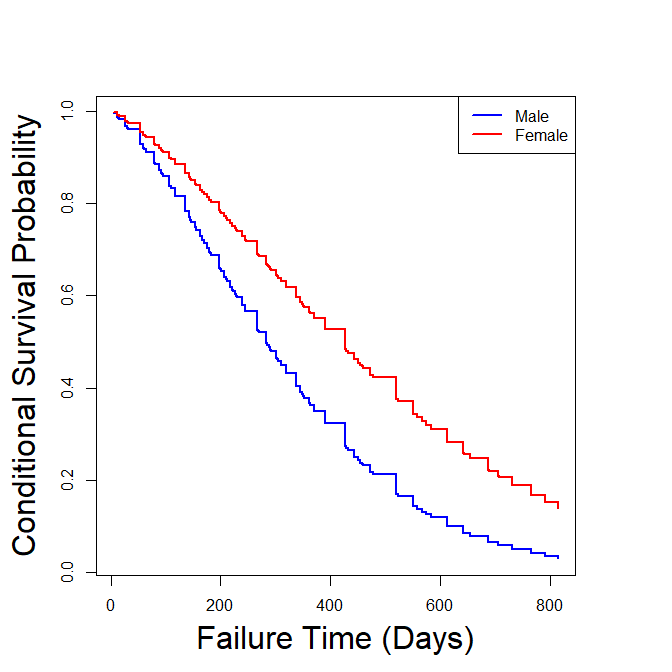}
}
\centering
\subfigure[ECOG]{
\includegraphics[width=0.42\linewidth]{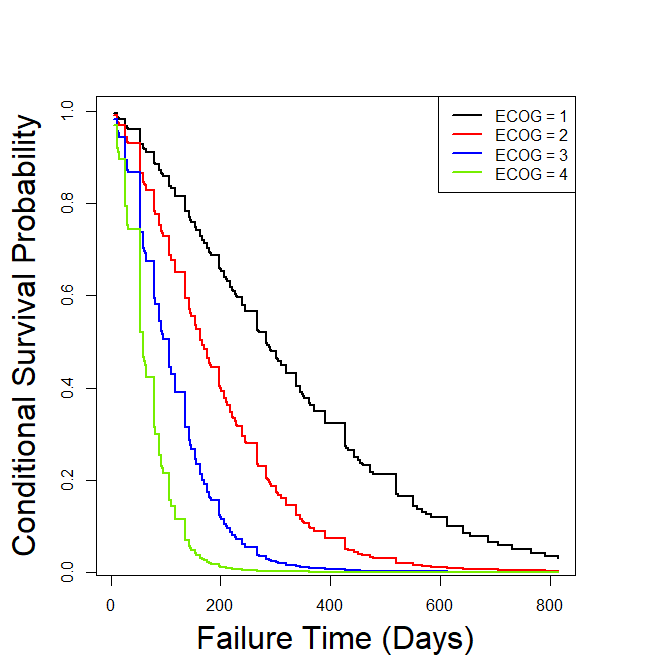}
}
\caption{Plots of the conditional survival probabilities based on SSNSM in the Lung cancer dataset for Sex and ECOG}
\label{survival_plot (lung)}
\end{figure}

\subsection{Breast cancer dataset}

For the second real data analysis, we explore data concerning breast cancer female patients sourced from the Surveillance, Epidemiology, and End Results (SEER) program of the National Cancer Institute (NCI), which offers comprehensive population-based cancer statistics. Our study aims to investigate the impact of covariates on the failure time specifically among individuals identified as Black in race.
This selected dataset comprises 291 patients, with a censoring rate of $74.9\%$. It includes some predictors such as Age, Size (indicating tumor size), Node examined (indicating the total number of regional lymph nodes that were removed and examined by the pathologist) and Node positive (indicating the exact number of regional lymph nodes examined by the pathologist that were found to contain metastases).
The dataset is available at https://dx.doi.org/10.21227/a9qy-ph35.

In Table \ref{tab:breast}, estimates from five different methods are displayed along with their corresponding standard error estimates computed using 500 bootstrap replicates. Overall, the estimates across the five methods exhibit similar values, except for the Age. In the case of SN and SSNSM, the Age estimates exhibit negative values, while estimates from the other methods are positive. The results for SSNSM show that the exact number of regional lymph nodes is significant to the failure time. Figure \ref{fig: survival_prob (breast)} depicts the conditional survival curves for selected values of the exact number of regional lymph nodes.  
Specifically, as the exact number of regional lymph nodes increases, the corresponding estimated survival function tends to decrease. 
The obtained IBS for Normal, SN, Gehan, GEE and SSNSM are $0.1814$, $0.1808$, $0.1913$, $0.1797$ and $0.1788$, respectively. SSNSM displays the best performance, closely followed by GEE. Note that we set $t_{\text{max}}$ as $90$ months.

\begin{table}[ht]
\caption{Estimated regression coefficients and standard errors in parentheses from the Breast cancer dataset}
\label{tab:breast} 
	\centering
	\begin{tabular}{ c  c c  c c   c     } 
        \hline \hline \noalign{\smallskip}
	Method	&
	Intercept & Age & Size & Node examined & Node positive  \\ \hline
     \multirow{2}{*}{Normal} &  5.2774 & 0.0012 & -0.0006 & 0.0288 & -0.0996  \\
       &  (0.7957) & (0.0134) & (0.0051) & (0.0188) & (0.0231)   \\
     \multirow{2}{*}{SN} & 5.1470  & -0.0024 & -0.0009 & 0.0280 & -0.0925  \\
       &  (0.7649) & (0.0140) & (0.0051) & (0.0184) & (0.0232)  \\
      \multirow{2}{*}{Gehan} &  5.9042 & 0.0011 & -0.0076 & 0.0468 & -0.1583   \\
      &  (1.6544) & (0.0230) & (0.0110) & (0.0341) & (0.0338) \\
    \multirow{2}{*}{GEE} & 5.1333 & 0.0024 & -0.0005 & 0.0265 & -0.0965   \\
       &  (0.7803) & (0.0127) & (0.0050) & (0.0181) & (0.0232)   \\
     \multirow{2}{*}{SSNSM} & 5.2884 & -0.0047 & -0.0007 & 0.0299 & -0.0954  \\
     &  (0.8312) & (0.0152) & (0.0053) & (0.0190) & (0.0250)   \\
    \noalign{\smallskip}\hline\noalign{\smallskip}
	\end{tabular}
\end{table}

\begin{figure}[ht] 
\centering
\includegraphics[width=0.45\linewidth]{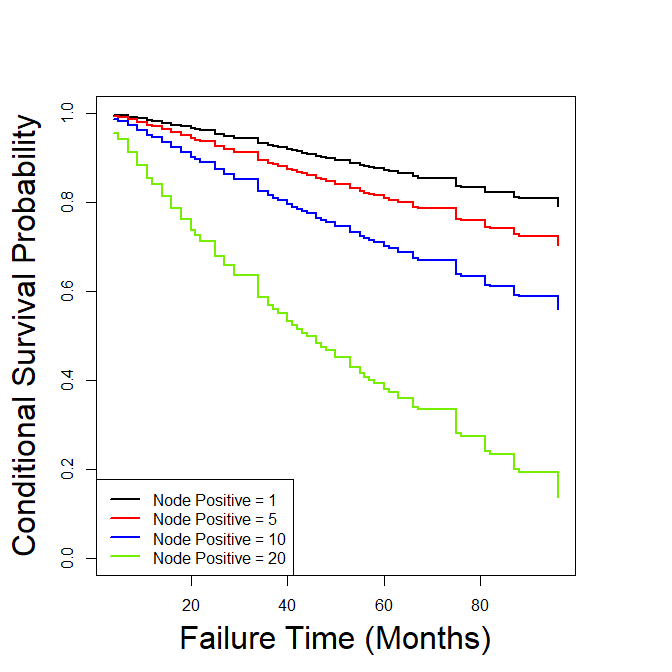}
\caption{Plots of the conditional survival probabilities based on SSNSM in the Breast cancer dataset for the exact number of regional lymph nodes}
\label{fig: survival_prob (breast)}
\end{figure}

\section{Discussion}
\label{sec7}

Ensuring the accuracy and efficiency of estimates in parametric AFT models requires addressing the potential misspecification issue stemming from incorrect assumptions about the error distribution. In our approach, we propose an AFT model that incorporates a semiparametric skew-normal scale mixture error distribution, offering robustness and flexibility. By adopting this model, we can obtain more reliable estimates even when the exact error distribution is unknown. 

While proposed model demonstrates robustness, it is important to note that the mean regression approach may not be suitable when the error distribution is asymmetric. To overcome this limitation, alternative methods have been proposed. For instance, \cite{galarza2017robust} proposed a quantile regression model utilizing specific members of the skew-normal scale mixture error distributions. Similarly, \cite{galarza2021skew} introduced a quantile regression model for censored and missing data using the skew-t error distribution. Additionally, \cite{zeng2022estimation} proposed the finite mixture of mode regression with a skew-normal distribution.

However, these approaches still rely on parametric assumptions, which can lead to misspecification issues. To address this concern, one possible solution is to consider quantile regression or modal regression models with semiparametric skew-normal scale mixture errors. We expect that these approaches may provide greater flexibility and effectively handle data with asymmetric distributions, reducing the risk of misspecification and improving the accuracy of the estimates.

\bibliographystyle{apa} 
\bibliography{references}

\section*{Appendix A: Proof of Theorem \ref{thm1}}
In \eqref{aft_ssnsm}, suppose that there exists $\tilde{b}_0$, $\tilde{\betab}$, $\tilde{\lambda}$ and $\tilde{Q}$ satisfying 
\begin{align}
    \int_{0}^{\infty} \frac{2}{\sigma} \phi\left(\frac{\log t - b_0 - \xb^{\top}\betab}{\sigma}\right)  \Phi\left(\lambda \frac{\log t - b_0 - \xb^{\top}\betab}{\sigma}\right) dQ(\sigma) \notag \\ 
    = \int_{0}^{\infty} \frac{2}{\sigma} \phi\left(\frac{\log t - \tilde{b}_0 - \xb^{\top}\tilde{\betab}}{\sigma}\right)  \Phi\left(\tilde{\lambda} \frac{\log t - \tilde{b}_0 - \xb^{\top}\tilde{\betab} }{\sigma}\right) d\tilde{Q}(\sigma). \label{ident}
\end{align}
The characteristic function of the left hand side of \eqref{ident} is 
\begin{align*}
\psi_{\log T | \Xb} (s) 
&=  \exp(i ( b_0 + \xb^{\top}\betab) s) \psi_{Q, \lambda} (s)
\end{align*}
where $i=\sqrt{-1}$, $z = \log t - b_0 - \xb^{\top}\betab$ and $\psi_{Q, \lambda} (s)$ is the characteristic function of the semiparametric skew-normal scale mixture distribution with the location $0$, latent distribution $Q$ and slant parameter $\lambda$.
The characteristic function of the right hand side of \eqref{ident} can be similarly represented as  
\begin{align*}
\tilde{\psi}_{\log T | \Xb} (t) &= \exp(i ( \tilde{b}_0 + \xb^{\top} \tilde{\betab}) s) \psi_{\tilde{Q}, \tilde{\lambda}} (s).
\end{align*}
Then the following equality must hold:
\begin{align}
\label{ident2}
 \exp(i (b_0 + \xb^{\top}\betab) s) \psi_{Q, \lambda} (s) =  \exp(i (\tilde{b}_0 + \xb^{\top} \tilde{\betab})  s) \psi_{\tilde{Q}, \tilde{\lambda}} (s). 
\end{align}

Let $U_1$ be an open set $U_{1} =\{\ub_{1} \in \mathbf{\mathcal{X}} | b_0 + \ub_{1}^{\top} \betab  \ne \tilde{b}_0 + \ub_{1}^{\top} \tilde{\betab} \}$, where $\mathbf{\mathcal{X}}$ is the support of $\xb$. Because $b_0 + \ub_{1}^{\top} \betab$ and $\tilde{b}_0 + \ub_{1}^{\top} \tilde{\betab}$ are continuous function for all $u_{1} \in U_{1}$, there exist open sets $V_{1}$ and $\tilde{V}_{1}$ in $\mathbb{R}$ satisfying
\begin{align*}
 \exp(i v_{1} s) \psi_{Q, \lambda} (s) =  \exp(i \tilde{v}_{1}  s) \psi_{\tilde{Q}, \tilde{\lambda}} (s),
\end{align*}
where $v_{1}\in V_{1}$ and $\tilde{v}_{1} \in \tilde{V}_{1}$.

Let $\kappa > 0$ satisfying $cv_{1} \in V_1$ and $c\tilde{v}_{1} \in \tilde{V}_{1}$ for all $c \in (1 - \kappa, 1 + \kappa)$. Then, the following equation is also hold. 
\begin{align}
\label{ident3}
\exp(i c v_{1}  s)  \psi_{Q, \lambda} (s) =  \exp(i c\tilde{v}_{1}  s) \psi_{\tilde{Q}, \tilde{\lambda}} (s).
\end{align}
Multiplying $\exp(-i c v_{1} s)$ on both sides of equation \eqref{ident3}, we have
\begin{align}
\label{ident4}
 \psi_{Q,\lambda} (s)
&= \exp(i c (\tilde{v}_{1} - v_{1})  s) \psi_{\tilde{Q},\tilde{\lambda}} (s)
\end{align}
for all $c \in (1 - \kappa, 1 + \kappa)$. Because the exponential function is analytic on the whole complex plane, if \eqref{ident4} holds for $c\in (1-\kappa,1+\kappa)$, \eqref{ident4} should also hold for all $-\infty<c<\infty$. Hence, the following equality should hold:
\begin{align}
\psi_{Q,\lambda} (s)
&=  \Bigg \{ \frac{1}{2C}\int_{-C}^C \exp(i c r_{1} ) dc \Bigg \} \psi_{\tilde{Q},\tilde{\lambda}} (s),
\end{align}
where $r_{1} = (\tilde{v}_{1} - v_{1}) s$, for all $C > 0$.

Letting $C \rightarrow \infty$, we can conclude $ \psi_{Q,\lambda} (s) = 0$ from 
\begin{align*}
    \lim_{C \to \infty} \frac{1}{2C} \int_{-C}^{C} \exp(ic r_{1}) dc 
        &= 0
\end{align*}
for any nonzero real number $r_1$.
However, letting $s \rightarrow 0$ implies $\psi_{Q,\lambda} (s) = 1$ which contradicts $ \psi_{Q,\lambda} (s) = 0$. 
As a result, we have $b_0 + \ub_{1}^{\top} \betab  = \tilde{b}_0 + \ub_{1}^{\top} \tilde{\betab}$.

Based on equation \eqref{ident2}, we can further derive that
\begin{align*}
 \exp(i b_0 s) \exp(i  \xb^{\top}\betab s) \psi_{Q, \lambda} (s) = \exp(i \tilde{b}_0  s) \exp(i \xb^{\top} \tilde{\betab}  s) \psi_{\tilde{Q}, \tilde{\lambda}} (s). 
\end{align*}
Suppose that $U_2$ is an open set as $U_{2} =\{\ub_{1} \in \mathbf{\mathcal{X}} | \ub_{1}^{\top} \betab  \ne \ub_{1}^{\top} \tilde{\betab} \}$. Within this set, we can find open sets $V_2$ and $\tilde{V}_2$ in $\mathbb{R}$ such that
\begin{align*}
\exp(i b_0 s) \psi_{Q, \lambda} (s) = \exp(i \tilde{b}_0 s) \psi_{\tilde{Q}, \tilde{\lambda}} (s) \Bigg \{ \frac{1}{2C}\int_{-C}^C \exp(i c r_{2} s) dc \Bigg \},
\end{align*}
where $r_2 = (\tilde{v}_{2} - v_{2}) s$, $v_{2}\in V_{2}$, $\tilde{v}_{2} \in \tilde{V}_{2}$, and $c \in (1 - \kappa, 1 + \kappa)$ satisfying $cv_{2} \in V_2$ and $c\tilde{v}_{2} \in \tilde{V}_2$. 
Letting $C \rightarrow \infty$, we obtain  
\begin{align}
\label{ident5}
\exp(i b_0 s) \psi_{Q, \lambda} (s) = 0.
\end{align}
However, the left hand side of \eqref{ident5} is $1$ if $s \rightarrow 0$, which contradicts the equation. 
Therefore, $(b_0, \betab^{\top})^{\top}$ must be $(\tilde{b}_0, \tilde{\betab}^{\top})^{\top}$. 
Because $(b_0, \betab^{\top})^{\top} = (\tilde{b}_0, \tilde{\betab}^{\top})^{\top}$, we can conclude that $\psi_{Q, \lambda} (s)$ is also equal to $ \psi_{\tilde{Q}, \tilde{\lambda}} (s)$.

Now, we need to verify that $\psi_{Q, \lambda} (s) = \psi_{\tilde{Q}, \tilde{\lambda}} (s)$ implies $Q = \tilde{Q}$ and $\lambda = \tilde{\lambda}$.
$\psi_{Q, \lambda} (s) $ is obtained as follows.
\begin{align*}
\psi_{Q, \lambda} (s) 
&=  \int_{0}^{\infty} \exp \Bigg(\frac{-s^2\sigma^2}{2}\Bigg) dQ(\sigma) + i  \int_{0}^{\infty} \exp \Bigg(\frac{-s^2\sigma^2}{2} \Bigg) e \Bigg( \frac{\Gamma \sigma s}{ \sqrt{2}} \Bigg ) dQ(\sigma), 
\end{align*}
where $\Gamma = \lambda / \sqrt{1 + \lambda^2}$ and $e (\cdot)$ is a complementary error function defined as
\begin{align*}
e(z) = - \frac{2i}{\sqrt{\pi}} \int_{0}^{iz} \exp(-t^2) dt.
\end{align*}
As the same way, 
\begin{align*}
\psi_{\tilde{Q}, \tilde{\lambda}} (s) &=  \int_{0}^{\infty} \exp\Bigg(\frac{-s^2\sigma^2}{2}\Bigg) d\tilde{Q}(\sigma) + i  \int_{0}^{\infty} \exp\Bigg(\frac{-s^2\sigma^2}{2}\Bigg) e \Bigg( \frac{\tilde{\Gamma} \sigma s}{\sqrt{2}} \Bigg) d\tilde{Q}(\sigma), 
\end{align*}
where  $\tilde{\Gamma} = \tilde{\lambda} / \sqrt{1 + \tilde{\lambda}^2}$.

Because $\psi_{Q, \lambda} (s) = \psi_{\tilde{Q}, \tilde{\lambda}} (s)$, we obtain
\begin{align}
\int_{0}^{\infty}\exp\Bigg(\frac{-s^2\sigma^2}{2}\Bigg) dQ(\sigma) &= \int_{0}^{\infty} \exp\Bigg(\frac{-s^2\sigma^2}{2}\Bigg) d\tilde{Q}(\sigma), \label{first} \\
\int_{0}^{\infty} \exp\Bigg(\frac{-s^2\sigma^2}{2}\Bigg) e \Bigg( \frac{\Gamma \sigma s}{ \sqrt{2}} \Bigg ) dQ(\sigma) &= \int_{0}^{\infty} \exp\Bigg(\frac{-s^2\sigma^2}{2}\Bigg) e \Bigg( \frac{\tilde{\Gamma} \sigma s}{\sqrt{2}} \Bigg) d\tilde{Q}(\sigma). \label{second}
\end{align}
The equation \eqref{first} implies that $Q = \tilde{Q}$ due to the uniqueness of the Laplace transform. 
Furthermore, in the \eqref{second}, $\Gamma$ is equal to $\tilde{\Gamma}$ because $e(\cdot)$ is one to one increasing function and $Q = \tilde{Q}$. Therefore, $\lambda = \tilde{\lambda}$ is also hold.

\section*{Appendix B: Proof of Theorem \ref{thm2}}

Let the density of $(\log Y = z, \delta)$ be 
    \begin{align*}
        p(z, \delta ; \bm{b}, \lambda, Q) = \Bigg \{ f(z ; \bm{b}, \lambda, Q) (1 - G(z) ) \Bigg \}^{\delta}  \Bigg \{g(z) S(z ; \bm{b}, \lambda, Q)  \Bigg \}^{1-\delta},
    \end{align*}
where $\bm{b} = (b_0, \betab)$ and $\hat{\bm{b}} = (\hat{b}_0, \hat{\betab})$.
Moreover, let us define the metric as
    \begin{align*}
        d \Bigg ( (\bm{b}, \lambda, Q), (\hat{\bm{b}}, \hat{\lambda}, \hat{Q})  \Bigg ) = \| \bm{b} - \hat{\bm{b}} \|_1 + | \lambda - \hat{\lambda} | + \int_{0}^{\infty} | Q (\sigma) - \hat{Q} (\sigma) | \exp(-|\sigma|) d\tau(\sigma),
    \end{align*}
where $\tau$ is the Lebesgue measure on $\mathbb{R}^{+}$.
\cite{kiefer1956consistency} proposed five assumptions, including continuity (Assumption 2), identifiability (Assumption 4) and integrability (Assumption 5), which guarantee that $(\hat{\bm{b}}, \hat{\lambda}, \hat{Q})$ converges in probability to $(\bm{b}, \lambda, Q)$, as denoted by $d \Bigg ( (\bm{b}, \lambda, Q), (\hat{\bm{b}}, \hat{\lambda}, \hat{Q}) \Bigg ) \overset{p}{\rightarrow} 0$.

The density of SSNSM in \eqref{aft_ssnsm} trivially satisfies the Assumption 1,2 and 3. Additionally, Assumption 4 is verified by Theorem \ref{thm1}.
Since $Q$ is supported by $[\ell, \infty)$, $p(z, \delta ; \bm{b}, \lambda, Q)$ is uniformly bounded, it implies that we just need to check if $- E \Bigg[\log p(Z, \delta ; \bm{b}, \lambda, Q)  \Bigg] < \infty$ for Assumption 5.
Since $\log p(z, \delta ; \bm{b}, \lambda, Q)$ can be derived as
    \begin{align*}
        \log p(z, \delta ; \bm{b}, \lambda, Q)  &= \delta \Bigg( \log f(z ; \bm{b}, \lambda, Q) + \log (1 - G(z) ) \Bigg)  \\
        &+ (1 - \delta) \Bigg( \log g(z) + \log S(z ; \bm{b}, \lambda, Q) \Bigg), 
    \end{align*}
$- E \Bigg[  \log p(Z, \delta) \Bigg ]$ can be written as
    \begin{align}
     - E \Bigg[  \log p(Z, \delta  ; \bm{b}, \lambda, Q) \Bigg ]  
     &=  - \int_{-\infty}^{\infty} \log \Bigg(f(z  ; \bm{b}, \lambda, Q)\Bigg)  f(z  ; \bm{b}, \lambda, Q) (1 - G(z)) dz  \label{1} \\
     &- \int_{-\infty}^{\infty}  \log \Bigg(1 - G(z)\Bigg)  f(z  ; \bm{b}, \lambda, Q) (1 - G(z)) dz \notag \\
     & - \int_{-\infty}^{\infty}  \log \Bigg( g(z) \Bigg) g(z) S(z  ; \bm{b}, \lambda, Q) dz \label{3} \\
     & - \int_{-\infty}^{\infty} \log \Bigg( S(z  ; \bm{b}, \lambda, Q) \Bigg) g(z) S(z  ; \bm{b}, \lambda, Q) dz \notag\\ 
     &\leq 2 \exp(-1) + \eqref{1} + \eqref{3} \notag,
    \end{align} 
since $0 \leq  -z \log z \leq \exp(-1)$ for $0 < z \leq 1$.

Now, we verify both \eqref{1} and \eqref{3} are bounded. Since the skew-normal distribution has a finite first moment, it also trivially has $\int_{-\infty}^{\infty} f(z ; \bm{b}, \lambda, \sigma) [ \log |z| ]^{+} dz < \infty $, where $f(z ; \bm{b}, \lambda, \sigma)$ is the density of the skew-normal distribution. As a result, we have $E\Bigg[\log | Z| \Bigg]^{+} < \infty$ by the condition $\int_{\ell}^{\infty} \log \sigma d Q(\sigma) < \infty$.
It implies that $- \int_{-\infty}^{\infty} \log \Bigg(f(z  ; \bm{b}, \lambda, Q)\Bigg)  f(z  ; \bm{b}, \lambda, Q) dz  < \infty$ by the lemma at the Section 2 in \cite{kiefer1956consistency}. Since $0 \leq 1-G(z) \leq 1$, \eqref{1} is also bounded. Furthermore, \eqref{3} is bounded by the condition $-\int_{-\infty}^{\infty} \log g(t) d G(t) < \infty$. Therefore, $ - E \Bigg[  \log p(Z, \delta  ; \bm{b}, \lambda, Q) \Bigg ] < \infty$.

\end{document}